\documentclass[11pt]{amsart}
\usepackage{mathrsfs}
\usepackage{amssymb}
\usepackage{amsmath}
\usepackage{amscd}
\newtheorem{theorem}{Theorem}[section]
\newtheorem{prop}[theorem]{Proposition}
\newtheorem{defn}[theorem]{Definition}
\newtheorem{lemma}[theorem]{Lemma}
\newtheorem{coro}[theorem]{Corollary}
\newtheorem{prop-def}{Proposition-Definition}[section]

\newtheorem{remark}[theorem]{Remark}

\newtheorem{exam}[theorem]{Example}

\begin{document}
\setlength{\oddsidemargin}{0cm} \setlength{\evensidemargin}{0cm}

\title{Pre-alternative algebras and pre-alternative bialgebras}

\author{Xiang Ni}

\address{Chern Institute of Mathematics \& LPMC, Nankai
University, Tianjin 300071, P.R.
China}\email{xiangn$_-$math@yahoo.cn}

\author{Chengming Bai}

\address{Chern Institute of Mathematics \& LPMC, Nankai University,
Tianjin 300071, P.R. China} \email{baicm@nankai.edu.cn}

\def\shorttitle{Pre-alternative algebras and pre-alternative bialgebras}

\begin{abstract}

We introduce a notion of pre-alternative algebra which may be seen
as an alternative algebra whose product can be decomposed into two
pieces which are compatible in a certain way. It is also the
``alternative" analogue of a dendriform dialgebra or a pre-Lie
algebra. The left and right multiplication operators of a
pre-alternative algebra give a bimodule structure of the associated
alternative algebra. There exists a (coboundary) bialgebra theory
for pre-alternative algebras, namely, pre-alternative bialgebras,
which exhibits all the familiar properties of the famous Lie
bialgebra theory. In particular, a pre-alternative bialgebra is
equivalent to a phase space of an alternative algebra and our study
leads to what we called $PA$-equations in a pre-alternative algebra,
which are analogues of the classical Yang-Baxter equation.
\end{abstract}

\subjclass[2000]{17D05, 17A30, 16W30 }

\keywords{Alternative algebra, pre-alternative algebra,
pre-alternative bialgebra, classical Yang-Baxter equation}

\maketitle

\tableofcontents \setcounter{section}{0}

\baselineskip=18pt

\section{Introduction}
\setcounter{equation}{0}
\renewcommand{\theequation}
{1.\arabic{equation}}

 A {\it dendriform dialgebra} is a vector
space $D$ together with two bilinear products denoted by
$\prec,\succ: D\otimes D\rightarrow D$ such that (for any $x,y,z\in
D$)
\begin{equation}
(x\prec y)\prec z=x\prec(y\circ z),\;(x\succ y)\prec z=x\succ(y\prec
z),\; (x\circ y)\succ z=x\succ(y\succ z),
\end{equation}
where $x\circ y=x\prec y+x\succ y$. The notion of dendriform
dialgebra was introduced by J.-L. Loday  in 1995 as
 the (Koszul) dual of the associative dialgebra, which is related to periodicity phenomenons
 in algebraic K-theory (\cite{L1}). It was further studied in connection with
 several areas in mathematics and physics, including operads (\cite{L3}), homology (\cite{Fra1}, \cite{Fra2}),
 Hopf algebras (\cite{Cha}, \cite{Ho}, \cite{Ron}), Lie and Leibniz algebras (\cite{Fra2}),
 combinatorics (\cite{AS1}, \cite{AS2}), arithmetic (\cite{L2}) and
 quantum field theory (\cite{EG}). We recommend \cite{L1} as a beautiful introduction and motivation
 of this subject.

 Moreover, for any dendriform dialgebra $(D,\prec,\succ)$, the bilinear product $\circ$ defines
 an associative algebra.
 Thus, a dendriform dialgebra may be seen as an associative algebra whose
 multiplication
 can be decomposed into two coherent operations.

Furthermore, we re-examine the identity (1.1) in the following way.
 Let $A$ be a vector space together with
two products denoted by $\prec,\succ: A\otimes A\rightarrow A$. The
{\it right associator (r-associator), middle associator
(m-associator)} and {\it left associator (l-associator)} are defined
respectively by
\begin{equation}
(x,y,z)_{r}=(x\prec y)\prec z-x\prec(y\circ z);
\end{equation}
\begin{equation}
(x,y,z)_{m}=(x\succ y)\prec z-x\succ(y\prec z);
\end{equation}
\begin{equation}
(x,y,z)_{l}=(x\circ y)\succ z-x\succ(y\succ z),\;\forall x,y,z\in A,
\end{equation}
where $x\circ y=x\prec y+x\succ y $. So $(D,\prec,\succ)$ is a
dendriform dialgebra if and only if all the above three
``associators" are zero.

 On the other hand, {\it alternative algebras} are a class of very important nonassociative
algebras (\cite{KS}, [S1-2]), which are closely related to {\it Lie
algebras} (\cite{S2}), {\it Jordan algebras} (\cite{J}) and {\it
Malcev algebras} (\cite{KS}), and so on. Motivated by the relations
between associative algebras and alternative algebras (see the
details in section 2), it is natural to consider certain algebra
structure on an alternative algebra as an analogue of an dendriform
dialgebra on an associative algebra. So we introduce a notion of
pre-alternative algebra, which is one of the main objects in this
paper. Just as an alternative algebra is a generalization of an
associative algebra which weakens the condition of associativity, a
pre-alternative algebra is a generalization of a dendriform
dialgebra which weakens the conditions of $l$-associativity,
$m$-associativity and $r$-associativity.

We would like to point out that there has already been a ``Lie
algebraic" version for relations between associative algebras and
dendriform dialgebras. That is, a class of nonassociative algebras,
namely pre-Lie algebras (or under other names like left-symmetric
algebras, Vinberg algebras and so on, see a survey article \cite{Bu}
and the references therein) play a similar role of dendriform
dialgebras. Therefore, in this sense, pre-alternative algebras are
just ``alternative" analogues of pre-Lie algebras or dendriform
dialgebras.

Furthermore, Goncharov constructed a bialgebra theory for
alternative algebras in \cite{G}, namely, alternative D-bialgebras.
In this paper, we show that there exists a (coboundary) bialgebra
theory for pre-alternative algebras, namely, pre-alternative
bialgebras, which exhibits all the familiar properties of the famous
Lie bialgebra theory of Drinfeld (\cite{D}). As well as an
alternative D-bialgebra is equivalent to an ``alternative analogue"
of Manin triple (\cite{G}, \cite{CP}), a pre-alternative bialgebra
is equivalent to a phase space of an alternative algebra (\cite{K1},
\cite{Bai1}). In particular there also exists a ``Drinfeld double"
construction for a pre-alternative bialgebra which is previously
unexpected. Moreover, There is also a clear analogue between
alternative D-bialgebras and pre-alternative bialgebras. On the
other hand, we would like to emphasis that the representation theory
 of alternative algebras and pre-alternative algebras  play
 an essential role in establishing the bialgebra theories.
We also would like to point out that both  alternative D-bialgebras
and pre-alternative
 bialgebras can be fit into the  general framework of
 {\it  generalized bialgebras} as introduced by Loday in \cite{L4}.
 So it would be interesting to
 find the relations with Loday's question, that is, to find {\it good triples of
 operads} (\cite{L4}).

The paper is organized as follows. In section 2, we study bimodules
of alternative algebras and introduce various methods to construct
pre-alternative algebras. In section 3, we generalize the notion of
phase space in mathematical physics (\cite{K1}) to the realm of
alternative algebras, and show that pre-alternative algebras are the
natural underlying structures.  In section $4$, we define and study
bimodules and matched pairs of pre-alternative algebras. In section
$5$, we introduce the notion of pre-alternative bialgebra which is
equivalent to a phase spaces of an alternative algebra. In section
$6$, we show that there is a reasonable coboundary (pre-alternative)
bialgebra theory and what we study leads to what we call
$PA$-equations. In section $7$, the properties of $PA$-equations are
discussed. In Appendix, we prove the main results in \cite{G} by a
little different approach and we point out that there is a
Drinfeld's double construction for an alternative $D$-bialgebra,
which was not given in \cite{G}.

Throughout this paper, all the algebras are finite-dimensional over
a fixed base field $\textbf{k}$ whose characteristic is not $2$. We
give some notations as follow.

(1) Let $V$ be a vector space.  Let $\mathfrak{B}:V\otimes
V\rightarrow\mathbb F$ be a symmetric or skew-symmetric bilinear
form on a vector space $V$. If $W$ is a subspace of $V$, then we
define
\begin{equation}
W^{\perp}=\{x\in V|\mathfrak{B}(x,y)=0,\forall y\in
W\}.\end{equation} If $W\subset W^{\perp}$ ($W=W^{\perp}$
respectively), then $W$ is called {\it isotropic} ({\it Lagrangian}
respectively).

(2) Let $(A,\diamond)$ be a vector space with a binary operation
$\diamond: A\otimes A\rightarrow A$. Let $\mathfrak{l}_\diamond(x)$
and $\mathfrak{r}_\diamond(x)$ denote the left and right
multiplication operator respectively, that is,
$\mathfrak{l}_\diamond(x)y=\mathfrak{r}_\diamond(y)x=x\diamond y$
for any $x,y\in A$. We also simply denote them by $\mathfrak{l}(x)$
and $\mathfrak{r}(x)$ respectively without confusion. Moreover, let
$\mathfrak{l}_\diamond, \mathfrak{r}_\diamond:A\rightarrow gl(A)$ be
two linear maps with $x\rightarrow \mathfrak{l}_\diamond(x)$ and
$x\rightarrow \mathfrak{r}_\diamond(x)$ respectively.

(3) Let $V$ be a vector space and $r=\sum_i{a_i\otimes b_i}\in
V\otimes V$. Set
\begin{equation}
r_{12}=\sum_ia_i\otimes b_i\otimes 1,\quad r_{13}=\sum_{i}a_i\otimes
1\otimes b_i,\quad r_{23}=\sum_i1\otimes a_i\otimes b_i,
\end{equation}
where $1$ is a scale. If in addition, there exists a binary
operation $\diamond: V\otimes V\rightarrow V$ on $V$, then the
operation between two $r$s is in an obvious way. For example,
\begin{equation}
r_{12}\diamond r_{13}=\sum_{i,j}a_i\diamond a_j\otimes b_i\otimes
b_j,\; r_{13}\diamond r_{23}=\sum_{i,j}a_i\otimes a_j\otimes
b_i\diamond b_j,\;r_{23}\diamond r_{12}=\sum_{i,j}a_j\otimes
a_i\diamond  b_j\otimes b_i.\end{equation} and so on.

(4) Let $V$ be a vector space. Let $\sigma: V\otimes V \rightarrow
V\otimes V$ be the {\it flip} defined as
\begin{equation}\sigma(x \otimes y) = y\otimes x,\quad
\forall x, y\in V.
\end{equation}
For any $r\in V\otimes V$, $r$ is called {\it symmetric} ({\it
skew-symmetric} respectively) if $r=\tau (r)$ ($r=-\tau(r)$
respectively). On the other hand, any $r\in V\otimes V$ can be
identified as a linear map $T_r: V^*\rightarrow V$ in the following
way:
\begin{equation}
\langle u^*\otimes v^*,\ r\rangle=\langle u^*,T_r(v^*)
\rangle,\quad \forall u^*, v^*\in V^*,
\end{equation}
where $\langle, \rangle$ is the canonical paring between $V$ and
$V^*$. $r\in V\otimes V$ is called {\it nondegenerate} if the above
induced linear map $T_r$ is invertible. Moreover, any invertible
linear map $T:V^*\rightarrow V$ can induce a nondegenerate bilinear
form $\mathfrak{B}( , )$ on $V$ by
\begin{equation}
\mathfrak{B}(u,v)=\langle T^{-1}u, v\rangle,\;\;\forall u,v\in V.
\end{equation}
Furthermore, $T$ is called {\it symmetric} ({\it skew-symmetric}
respectively) if the induced bilinear form ${\mathfrak{B}}$ is
symmetric (skew-symmetric respectively). Obviously, the symmetry or
skew-symmetry of both $T$ and its corresponding $r\in V\otimes V$
coincides.

(5) Let $V_1,V_2$ be two vector spaces and $T:V_1\rightarrow V_2$ be
a linear map. Denote the dual (linear) map by $T^*:V_2^*\rightarrow
V_1^*$ defined by
\begin{equation}
\langle  v_1,T^*(v_2^*)\rangle   =\langle  T(v_1),v_2^*\rangle
,\;\;\forall v_1\in V_1, v_2^*\in V_2^*.
\end{equation}
On the other hand,  $T$ can be identified as an element $r_T\in
V_2\otimes V_1^*$ by
\begin{equation}
\langle r_T, v_2^*\otimes v_1\rangle =\langle T(v_1),
v_2^*\rangle,\;\;\forall v_1\in V_1, v_2^*\in V_2^* .\end{equation}
Note that equation (1.9) is exactly the case that $V_1=V_2^*$.
Moreover, in the above sense, any linear map $T:V_1\rightarrow V_2$
is obviously  an element in $(V_2\oplus V_1^*)\otimes (V_2\oplus
V_1^*)$.

(6) Let $A$ be an algebra and $V$ be a vector space. For any linear
map $\rho: A\rightarrow gl(V)$, define a linear map $\rho^* : A
\rightarrow gl(V^*)$ by
\begin{equation}
\langle \rho^*(x)v^*, u\rangle = \langle v^*, \rho(x)u\rangle,\
\forall x \in A, u\in V, v^*\in V^*.
\end{equation}
Note that in this case, $\rho^*$ is different from the one given by
equation (1.11) which regards $gl(V)$ as a vector space, too.

(7) Let $V_1$ and $V_2$ be two vector spaces, we denote the elements
of $V_1\oplus V_2$ by $u+v$ or $(u,v)$ for $u\in V_1, v\in V_1$.

(8) Let $V$ be a vector space, we sometimes use 1 to denote the
identity transformation of $V$.

\section{Representation theory of alternative algebras and pre-alternative algebras}

\setcounter{equation}{0}
\renewcommand{\theequation}
{2.\arabic{equation}}

\begin{defn} {\rm  An {\it alternative algebra} $(A,\circ)$ is a vector space $A$ equipped with a
bilinear operation $(x,y)\rightarrow x\circ y$ satisfying
\begin{equation}
(x,x,y)=(y,x,x)=0,\;\;\forall x,y,z\in A,
\end{equation}
where $(x,y,z)=(x\circ y)\circ z-x\circ(y\circ z)$ is the {\it
associator}.} \end{defn}

\begin{remark}{\rm If the characteristic of the field is not 2, then an alternative
algebra $(A,\circ)$ also satisfies the stronger axioms:
\begin{equation}
(x_1,x_2,y)+(x_2,x_1,y)=0,\;\;(y,x_1,x_2)+(y,x_2,x_1)=0,\;\;\forall
x_1,x_2, y\in A.
\end{equation}
}
\end{remark}

\begin{defn}{\rm (\cite{S1})\quad Let $(A,\circ)$ be an alternative algebra and $V$ be a vector
space. Let $L,R: A\rightarrow gl(V)$ be two linear maps. $V$ (or the
pair $(L,R)$, or $(V,L,R)$) is called a {\it representation} or a
{\it bimodule} of $A$ if (for any $x,y\in A$)
\begin{equation}
L(x^2)=L(x)L(x),\;\; R(x^2)=R(x)R(x),
\end{equation}
\begin{equation}
R(y)L(x)-L(x)R(y)=R(x\circ y)-R(y)R(x),\; L(y\circ
x)-L(y)L(x)=L(y)R(x)-R(x)L(y).
\end{equation}}
\end{defn}

According to \cite{S3}, $(L,R)$ is a bimodule of an alternative
algebra $(A,\circ)$ if and only if the direct sum $A\oplus V$ of
vector spaces is turned into an alternative algebra (the {\it
semidirect sum}) by defining multiplication in $A\oplus V$ by
\begin{equation}
(x_1+v_1)\ast(x_2+v_2)=x_1\circ x_2+(L(x_1)v_2+R(x_2)v_1),
\;\;\forall x_1,x_2\in A,v_1, v_2\in V. \end{equation}
 We denote it by $A\ltimes_{L,R}V$ or simply $A\ltimes V$.

\begin{prop}
 If $(V,L,R)$ is a bimodule of an
alternative algebra $(A,\circ)$, then $(V^*,R^*,L^*)$ is a bimodule
of $(A,\circ)$.
\end{prop}

\begin{proof}
By equations (2.3) and (2.4), we have
$$L(x\circ y)-L(x)L(y)=-L(y\circ x)+L(y)L(x)=R(x)L(y)-L(y)R(x),\;\;\forall x,y\in A.$$
So for any $u^*\in A^*,v\in V$, we have

{\small \begin{eqnarray*} \langle
(L^*(y)R^*(x)-R^*(x)L^*(y))u^*,v\rangle&=&\langle
u^*,(R(x)L(y)-L(y)R(x))v\rangle=\langle u^*,(L(x\circ
y)-L(x)L(y))v\rangle\\
&=&\langle(L^*(x\circ y)-L^*(y)L^*(x))u^*,v\rangle.
\end{eqnarray*}}
Hence $L^*(y)R^*(x)-R^*(x)L^*(y)=L^*(x\circ y)-L^*(y)L^*(x)$.
Similarly, $(R^*,L^*)$ also satisfies the other axioms defining a
bimodule of $(A,\circ)$.
\end{proof}

\begin{defn} {\rm A {\it pre-alternative algebra} $(A,\prec,\succ)$ is a vector space
$A$ with two bilinear products denoted by $\prec,\succ: A\otimes
A\rightarrow A$ satisfying
\begin{equation}
(x,y,z)_{m}+(y,x,z)_r=0,(x,y,z)_{m}+(x,z,y)_{l}=0,
(y,x,x)_{r}=(x,x,y)_{l}=0,\forall x,y,z\in A,\end{equation} where
$(x,y,z)_r,(x,y,z)_m,(x,y,z)_l$ are defined by equations (1.2)-(1.4)
respectively and $x\circ y=x\succ y+x\prec y$.}
 \end{defn}

\begin{remark}
{\rm (1) If the characteristic of the field is not 2, then a
pre-alternative algebra $(A,\prec$, $\succ)$ satisfies the strong
axioms (for any $x,y,z\in A$):
\begin{equation}
(x,y,z)_{m}+(y,x,z)_r=0,\;\;
(x,y,z)_{m}+(x,z,y)_{l}=0,\end{equation}
\begin{equation}
(x,y,z)_l+(y,x,z)_l=0,\;\;(x,y,z)_r+(x,z,y)_r=0.\end{equation}

(2) It would be interesting to describe free pre-alternative
algebras (cf. \cite{L1}). }
\end{remark}

\begin{prop} Let $(A,\prec,\succ)$ be a pre-alternative algebra.
Then the product
\begin{equation}
x\circ y=x\succ y+x\prec y,\;\;\forall x,y\in A,
\end{equation}
defines an alternative algebra, which is called the {\it associated}
alternative algebra of $A$ and denoted by $As(A)$. $(A,\prec,\succ)$
is called a {\it compatible} pre-alternative algebra structure on
the alternative algebra $As(A)$.
\end{prop}

\begin{proof}
In fact, for any $x,y\in A$, we have

{\small\begin{eqnarray*}
(x,x,y)&=&(x\circ x)\circ y-x\circ(x\circ y)\\
&=&(x\circ x)\succ y+(x\succ x)\prec y+(x\prec x)\prec
y-x\succ(x\succ y)-x\succ(x\prec y)-x\prec(x\circ y)\\
&=&(x,x,y)_l+(x,x,y)_m+(x,x,y)_r=0.
\end{eqnarray*}}
Similarly, we  show that $(y,x,x)=0$.
\end{proof}

\begin{remark}
{\rm Thus, a pre-alternative algebra can be seen as an alternative
algebra whose product can be decomposed into two pieces which are
compatible in a certain way. On the other hand, it is obvious that
an associative algebra is an alternative algebra and a dendriform
dialgebra is a pre-alternative algebra.}

\end{remark}

If $(A,\circ)$ is an alternative algebra, then
$(A,\mathfrak{l}_{\circ},\mathfrak{r}_{\circ})$ is a bimodule of
$A$. Moreover,

\begin{prop}Let $(A,\prec,\succ)$ be a pre-alternative algebra.
Then $(A,\mathfrak{l}_{\succ},\mathfrak{r}_{\prec})$ is a bimodule
of the associated alternative algebra $(As(A),\circ)$.
\end{prop}

\begin{proof}
In fact, for any $x,y,z\in A$, we have
\begin{eqnarray*}
(\mathfrak{r}_{\prec}(y)\mathfrak{l}_{\succ}(x)-\mathfrak{l}_{\succ}(x)\mathfrak{r}_{\prec}(y))z&=&(x\succ
z)\prec y-x\succ(z\prec y)=z\prec(x\circ y)-(z\prec x)\prec
y\\&=&(\mathfrak{r}_{\prec}(x\circ
y)-\mathfrak{r}_{\prec}(y)\mathfrak{r}_{\prec}(x))z.
\end{eqnarray*}
Similarly, $(\mathfrak{l}_{\succ},\mathfrak{r}_{\prec})$ satisfies
the other axioms defining a bimodule of $(As(A),\circ)$.
\end{proof}

Let $(A,\prec,\succ)$ be a pre-alternative algebra. Then by
Propositions 2.4 and  2.9, both
$(A^*,\mathfrak{r}_{\circ}^*,\mathfrak{l}_{\circ}^*)$ and
$(A^*,\mathfrak{r}_{\prec}^*,\mathfrak{l}_{\succ}^*)$ are bimodules
of the associated alternative algebra $As(A)$.

The following terminology is motivated by the notion of $\mathcal
O$-operator as a generalization of (the operator form of) the
classical Yang-Baxter equation in \cite{K2} (also see \cite{Bai2}).

\begin{defn}{\rm Let $(V,L,R)$ be a bimodule of an alternative algebra
$(A,\circ)$. A linear map $Al: V\rightarrow A$ is called an {\it
Al-operator} associated to $(V,L,R)$ if
\begin{equation}
Al(u)\circ Al(v)=Al(L(Al(u))v+R(Al(v))u),\;\forall u,v\in V.
\end{equation}}
\end{defn}

\begin{prop} Let $Al : V \rightarrow A $ be an $Al$-operator of an
alternative algebra  $(A,\circ)$ associated to a bimodule $(V,L,R)$.
Then there exists a pre-alternative algebra structure on V given by
\begin{equation}
u\prec v=R(Al(v))u,\;\;  u\succ v=L(Al(u))v,\quad \forall u,v\in V.
\end{equation}
Therefore $V$ is an alternative algebra as the associated
alternative algebra of this pre-alternative algebra and $T$ is a
homomorphism of alternative algebras. Furthermore,
$Al(V)=\{Al(v)|v\in V\}\subset A$ is an alternative subalgebra of
$(A,\circ)$ and there is an induced pre-alternative algebra
structure on $T(V)$ given by
\begin{equation}
Al(u)\prec Al(v) = Al(u\prec v),\;\; Al(u)\succ Al(v) = Al(u\succ
v),\;\; \forall u,v\in V.
\end{equation}
Moreover, its associated alternative algebra structure is just the
alternative subalgebra structure of $(A,\circ)$ and $Al$ is a
homomorphism of pre-alternative algebras.
\end{prop}

\begin{proof}
We only prove one identity for $(V,\prec,\succ)$ being a
pre-alternative algebra as an example, the proof of the others is
similar. In fact, for any $u,v,w\in V$,
$$(u\succ v)\prec w+(v\prec u)\prec w=R(Al(w))L(Al(u))v+R(Al(w))R(Al(u))v,$$
$$u\succ(v\prec w)+v\prec(u\circ v)=L(Al(u))R(Al(w))v+R(Al(u\circ
w))v.$$ By equations (2.4), (2.10) and (2.11), we show that
$$(u,v,w)_m+(v,u,w)_r=(u\succ v)\prec w+(v\prec u)\prec w-u\succ(v\prec w)-v\prec(u\circ v)
=0.$$ The remaining parts of the conclusion are obvious.
\end{proof}

\begin{defn}{\rm (\cite{S1})\quad Let $(A,\circ)$ be an alternative
algebra and $(V,L,R)$ be a bimodule. A {\it $1$-cocycle} of $A$ into
$V$ is a linear map $D:A\rightarrow V$ satisfying
\begin{equation}
D(x\circ y)=L(x)D(y)+R(y)D(x),\;\; \forall x,y\in A.
\end{equation}}
\end{defn}

\begin{prop} Let $(A,\circ)$ be an alternative
algebra. Then the following conditions are equivalent.

(1) There exists a compatible pre-alternative algebra structure
$(A,\prec,\succ)$ on $(A,\circ)$.

(2) There exists an invertible $Al$-operator.

(3) There exists a bijective $1$-cocycle.

\end{prop}

\begin{proof} (3) $\Longrightarrow$ (2)\quad If $D$ is a bijective
$1$-cocycle of $(A,\circ)$ into a bimodule $(V,L,R)$, then $D^{-1}$
is an $Al$-operator associated to $(V,L,R)$.

(2)$\Longrightarrow$ (1)\quad  If $Al:V\rightarrow A$ is an
invertible $Al$-operator associated to a bimodule $(V, L,R)$, then
there is a compatible pre-alternative algebra structure on $A$ given
by
\begin{equation}
x\prec y=Al(R(y)Al^{-1}(x)),\;\;  x\succ y=Al(L(x)Al^{-1}(y)),\;\;
\forall x,y\in A.
\end{equation}

(1) $\Longrightarrow$ (3)\quad If $(A,\prec,\succ)$ is a compatible
pre-alternative algebra structure on $(A,\circ)$, then it is obvious
that the identity map ${\rm id}$ is a bijective $1$-cocycle of $A$
into the bimodule $(A,\mathfrak{l}_{\succ}, \mathfrak{r}_{\prec})$.
\end{proof}

\begin{exam}{\rm  Let $(A,\circ)$ be an alternative algebra
graded by positive integers, that is, $A=\oplus_{i\in{\rm N}}A_i$,
$A_i\circ A_j\subset A_{i+j}$. Then there is a bijective $1$-cocycle
associated to the bimodule
$(A,\mathfrak{l}_{\circ},\mathfrak{r}_{\circ})$ defined by
$D(x_i)=ix_i$ for $x_i\in A_i$. Therefore there exists a compatible
pre-alternative algebra structure on $(A,\circ)$ given by
$$x_i\succ
x_j={j\over{i+j}}x_i\circ x_j,\;\;x_i\prec x_j={i\over{ i+j}
}x_i\circ x_j,\;\;\forall x_i\in A_i,x_j\in A_j.$$}
\end{exam}

\begin{defn} {\rm Let $(A,\circ)$ be an arbitrary algebra which is not necessarily associative
 and
$\omega$ be a skew symmetric bilinear form on $A$. The bilinear form
$\omega$ is called {\it closed} if $\omega$ satisfies
\begin{equation}
\omega(a\circ b,c)+\omega(b\circ c,a)+\omega(c\circ a,b)=0,\;\;
\forall a,b,c\in A. \end{equation} If in addition, $\omega$ is
nondegenerate, then $\omega$ is called {\it symplectic}. In
particular, an alternative algebra $A$ equipped with a symplectic
form is called a {\it symplectic alternative algebra}.}
\end{defn}

\begin{prop}
 Let ($A,\circ,\omega$) be an
alternative algebra with a symplectic form $\omega$. Then there
exists a compatible pre-alternative algebra structure
``$\prec,\succ$" on $A$ given by
\begin{equation}\omega(x\prec y,z)=\omega(x,y\circ z),\;\;
 \omega(x\succ y,z)=\omega(y,z\circ x),\;\;\forall x, y,z\in A. \end{equation}\end{prop}

\begin{proof} Define a linear map $T:
A\rightarrow A^*$ by $\langle T(x), y\rangle=\omega(x, y), \forall
x, y\in A$. Then $T$ is invertible and $T$ is a 1-cocycle of $A$
into the bimodule $(A^*,\mathfrak{r}_\circ^*,
\mathfrak{l}_\circ^*)$. So by Proposition 2.13, there is a
compatible pre-alternative algebra structure ``$\prec,\succ$" on $A$
given by
$$x\prec y=T^{-1}(\frak{l}_{\circ}^*(y)T(x)),\quad x\succ
y=T^{-1}(\frak{r}_{\circ}^*(x)T(y)),\;\;\forall x,y\in A.$$ Thus,
for any $x,y\in A$,
$$\omega(x\prec y,z)=\langle T(x\prec
y),z\rangle=\langle\frak{l}_{\circ}^*(y)T(x),z\rangle=\omega(x,y\circ
z),$$
$$\omega(x\succ y,z)=\langle T(x\succ y),z\rangle=\langle
\frak{r}_{\circ}^*(x)T(y),z\rangle=\omega(y,z\circ x).$$ Therefore
the conclusion holds.
\end{proof}

\section{Alternative D-bialgebras and an alternative analog of the
classical Yang-Baxter equation}

\setcounter{equation}{0}
\renewcommand{\theequation}
{3.\arabic{equation}}

\begin{defn} {\rm (\cite{G},\cite{Z})\quad  Let $M$ be an arbitrary variety of
$\textbf{k}$-algebras and $(A,\circ)$ be an algebra in $M$ with
comultiplication $\bigtriangleup$. Then $(A,\circ,\bigtriangleup)$
is called an {\it M-bialgebra} in the sense of Drinfeld if $D(A)$
belongs to $M$, where $D(A)=A\oplus A^*$ is equipped with the
multiplication
\begin{equation}
(a+f)\star(b+g)=(a\circ b+f\cdot b+a\cdot g)+(f\ast g+f\bullet
b+a\bullet g),\;\;\forall a,b\in A, f,g\in A^*,\end{equation} where
$f\cdot a=\sum_a{a_{(1)}\langle f,a_{(2)}\rangle },a\cdot
f=\sum_a{\langle f,a_{(1)}\rangle a_{(2)}},\langle f\bullet
a,b\rangle =\langle f,a\circ b\rangle , \langle a\bullet f,b\rangle
=\langle f,b\circ a\rangle$,
$\bigtriangleup(a)=\sum_a{a_{(1)}\otimes a_{(2)}}$, and the
multiplication $\ast$ on $A^*$ is induced by $\bigtriangleup$. In
this case, $D(A)=A\oplus A^*$ is called the {\it Drinfeld's double}
for $A$. In particular, when $M$ is a variety of alternative
algebras, $(A,\circ,\bigtriangleup)$ is called an {\it alternative
D-bialgebra}.}
\end{defn}

\begin{remark}{\rm
According to \cite{G}, an alternative $D$-bialgebra
$(A,\circ,\Delta)$ is equivalent to an ``alternative analogue" of
Manin triple (\cite{CP}): there is an alternative algebra structure
on the direct sum $A\oplus A^*$ of the underlying vector spaces of
$A$ and $A^*$ such that both $A$ and $A^*$ are subalgebras and the
symmetric bilinear form on $A\oplus A^*$ given by
\begin{equation}{\mathcal B}(x+a^*, y+b^*)=\langle a^*, y\rangle
+\langle x, b^*\rangle,\;\;\forall x,y\in A, a^*,b^*\in
A^*,\end{equation} is invariant. Recall a bilinear form $\mathcal B$
on an alternative algebra $(A,\circ)$ is called {\it invariant} if
\begin{equation}
\mathcal B(x\circ y,z)=\mathcal B(x, y\circ z),\quad \forall
x,y,z\in A.
\end{equation}
On the other hand, it is easy to show that an alternative
$D$-bialgebra $(A,\circ,\Delta)$ is equivalent to the fact that
$(A,A^*,\mathfrak{r}_{\circ}^*,\mathfrak{l}_{\circ}^*,\mathfrak{r}_{\ast}^*,\mathfrak{l}_{\ast}^*)$
is a matched pair of alternative algebras in the sense of
Proposition 4.7.}
\end{remark}

\begin{defn}
 {\rm Let $(A,\circ)$ be an alternative algebra and
$r=\sum_i{a_i\otimes b_i}\in A\otimes A$. Then the pair $(A,r)$ is
called a {\it coboundary alternative D-bialgebra} if
$(A,\circ,\bigtriangleup_r)$, where
\begin{equation}
\bigtriangleup_r(x)=\sum_i{a_i\circ x\otimes b_i}-\sum_i{a_i\otimes
x\circ b_i},\quad \forall x\in A,\end{equation} is an alternative
$D$-bialgebra.}
\end{defn}

\begin{theorem}
{\rm(\cite{G})}\quad  Let $(A,\circ)$ be an alternative algebra and
$r\in A\otimes A$. Assume that $r$ is skew-symmetric and
\begin{equation}
C_A(r)=r_{23}\circ r_{12}-r_{12}\circ r_{13}-r_{13}\circ r_{23}=0.
\end{equation}
Then $(A,\circ,\bigtriangleup_r)$ is an alternative $D$-bialgebra.
\end{theorem}

\begin{defn} {\rm Let $(A,\circ)$ be an alternative algebra and
$r\in A\otimes A$. Equation (3.5) is called an {\it ``alternative
analog" of the classical Yang-Baxter equation} in \cite{G}. We also
call it the {\it alternative Yang-Baxter equation} in $(A,\circ)$.}
\end{defn}

\begin{prop} Let $(A,\circ)$ be an alternative algebra and
$r\in A\otimes A$. Then $r$ is a skew-symmetric solution of the
alternative Yang-Baxter equation in $(A,\circ)$ if and only if $T_r$
is an $Al$-operator associated to the bimodule
$(A^*,\mathfrak{r}_{\circ}^*,\mathfrak{l}_{\circ}^*)$, that is,
$T_r$ satisfies following equation \begin{equation}T_r(a^*)\circ
T_r(b^*)=T_r(\mathfrak{r}_{\circ}^*(T_r(a^*))b^*+
\mathfrak{l}_{\circ}^*(T_r(b^*))a^*),\;\; \forall a^*,b^*\in
A^*.\end{equation} So there is a pre-alternative algebra structure
on $A^*$ given by \begin{equation} a^*\prec
b^*=\mathfrak{l}_{\circ}^*(T_r(b^*))a^*,\;\; a^*\succ
b^*=\mathfrak{r}_{\circ}^*(T_r(a^*))b^*, \;\; \forall a^*,b^*\in
A^*.\end{equation} Moreover, its associated alternative algebra
structure is exactly the alternative algebra structure on $A^*$ as a
subalgebra of $D(A)=A\oplus A^*$ which is induced from the
comultiplication defined by equation (3.4). We denote this
alternative algebra structure on $A^*$ by $A^*(r)$.
\end{prop}

\begin{proof}
 Let $\{e_i,...,e_n\}$ be a basis of $A$ and
$\{e_i^*,...,e_n^*\}$ be its dual basis. Suppose that $e_i\circ
e_j=\sum\limits_kc_{ij}^ke_k$ and $r=\sum_{i,j}{a_{ij}e_i\otimes
e_j}$. Hence $a_{ij}=-a_{ji}$ and
$T_r(e_l^*)=\sum\limits_ka_{kl}e_k$. Then $r$ is a solution of the
alternative Yang-Baxter equation in $(A,\circ)$ if and only if (for
any $i,k,t$)
$$\sum_{js}{a_{st}a_{ij}c_{sj}^k-a_{jk}a_{st}c_{js}^i-a_{ij}a_{ks}c_{js}^t}=0.$$
The left hand side of the above equation is precisely the
coefficient of $e_i$ in $$-T_r(e_k^*)\circ
T_r(e_t^*)+T_r(\mathfrak{r}_{\circ}^*(T_r(e_k^*))e_t^*+\mathfrak{l}_{\circ}^*(T_r(e_t^*))e_k^*).$$
Thus the first half part of the conclusion holds.  It is easy to get
the other results.\end{proof}

\begin{coro} Let $(A,\circ)$ be an alternative algebra
and $r\in A\otimes A$. Assume $r$ is skew-symmetric and there exists
a nondegenerate symmetric invariant bilinear form $\mathcal B$ on
$(A,\circ)$. Define a linear map $\varphi:A\rightarrow A^*$ by
$\langle \varphi(x),y\rangle =\mathcal B(x,y)$ for any $x,y\in A$.
Then $r$ is a solution of  the alternative Yang-Baxter equation in
$(A,\circ)$
 if
and only if $\tilde{T}_r=T_r\varphi: A\rightarrow A$ is an
$Al$-operator associated to the bimodule
$(A,\mathfrak{l}_{\circ},\mathfrak{r}_{\circ})$, that is,
$\tilde{T}_r$ satisfies the following equation
\begin{equation}
\tilde{T}_r(x)\circ\tilde{T}_r(y)=\tilde{T}_r(\tilde{T}_r(x)\circ
y+x\circ\tilde{T}_r(y)),\;\;\forall x,y\in A.
\end{equation}
Hence there is a pre-alternative algebra structure on $A$ given by
\begin{equation}x\prec y=x\circ \tilde{T}_r(y),\;\; x\succ y=\tilde{T}_r(x)\circ y,\;\;\forall x,y\in A.
\end{equation}
\end{coro}

\begin{proof} For any $x,y\in A$, we have
$$\langle \varphi(\mathfrak{l}_\circ(x) y), z\rangle={\mathcal B}(x\circ y, z)={\mathcal B}(z,x\circ
y)={\mathcal B}(y, z\circ x)=\langle
\mathfrak{r}_\circ^*(x)\varphi(y),\ z\rangle, \;\;\forall x, y, z\in
A.$$ Hence $\varphi(\mathfrak{l}_\circ(x) y)=
\mathfrak{r}_\circ^*(x)\varphi(y)$ for any $x, y\in A$. Similarly,
$\varphi(\mathfrak{r}_\circ(x) y)=
\mathfrak{l}_\circ^*(x)\varphi(y)$ for any $x, y\in A$. Let
$a^*=\varphi(x), b^*=\varphi(y)$, then by Proposition 3.6, $r$ is a
solution of the alternative Yang-Baxter equation in $(A,\circ)$ if
and only if
$$
T_r\varphi(x)\circ T_r\varphi(y)=T_r(a^*)\circ
T_r(b^*)=T_r(\mathfrak{r}_{\circ}^*(T_r(a^*))b^*+
\mathfrak{l}_{\circ}^*(T_r(b^*))a^*)=T_r\varphi(T_r\varphi(x)\circ
y+x\circ T_r\varphi(y)).
$$
Therefore the conclusion holds.\end{proof}

\begin{remark}{\rm Equation (3.8) is exactly the {\it Rota-Baxter
relation of weight zero} for an alternative algebra (cf. \cite{Bax},
\cite{R}).}
\end{remark}

\begin{prop}
Let $(A,\circ)$ be an alternative algebra. Let $(V,L,R)$ be a
bimodule of $A$ and $(V^*,R^*,L^*)$ be its dual bimodule. Let
$T:V\rightarrow A$ be a linear map which can be identified as an
element in $A\ltimes_{R^*,L^*}V^*\otimes A\ltimes_{R^*,L^*}V^*$.
Then $T$ is an $Al$-operator of $A$ associated to $(V,L,R)$ if and
only if $r=T-\tau T$ is  a skew-symmetric solution of the
alternative Yang-Baxter equation in $A\ltimes_{R^*,L^*}V^*$.
\end{prop}

\begin{proof} Let $\{e_1,\cdots,e_n\}$ be a basis of $A$. Let
$\{v_1,\cdots, v_m\}$ be a basis of $V$ and $\{ v_1^*,\cdots,
v_m^*\}$ be its dual basis. Set
$T(v_i)=\sum\limits_{k=1}^na_{ik}e_k, i=1,\cdots, m$. Then
$$T=\sum_{i=1}^m T(v_i)\otimes v_i^*=\sum_{i=1}^m\sum_{k=1}^n
a_{ik}e_k\otimes v_i^*\in A\otimes V^*\subset
(A\ltimes_{R^*,L^*}V^*)\otimes (A\ltimes_{R^*,L^*}V^*).$$ Therefore
we have {\small{\begin{eqnarray*} r_{12}\circ r_{13}
&=&\sum_{i,k=1}^m\{T(v_i)\circ T(v_k)\otimes v_i^*\otimes
v_k^*-R^*(T(v_i))v_k^*\otimes v_i\otimes T(v_k)-L^*(T(v_k))v_i^*\otimes T(v_i)\otimes v_k^* \};\\
r_{23}\circ r_{12} &=&\sum_{i,j=1}^m\{T(v_k)\otimes R^*(T(v_i))v_k^*\otimes v_i^*-
v_k^*\otimes T(v_i)\circ T(v_k)\otimes v_i^*+v_k^*\otimes L^*(T(v_k))v_i^*\otimes T(v_i)\};\\
r_{13}\circ r_{23} &=&\sum_{i,k=1}^m\{v_i^*\otimes v_k^*\otimes
T(v_i)\circ T(v_k)-T(v_i)\otimes v_k^*\otimes
L^*(T(v_k))v_i^*-v_i^*\otimes T(v_k)\otimes R^*(T(v_i))v_k^*\}.
\end{eqnarray*}}}
By the definition of the dual bimodule, we know
$$L^*(T(v_k))v_i^*=\sum_{j=1}^m\langle v_i^*, L(T(v_k))v_j\rangle v_j^*,\;\;
  R^*(T(v_k))v_i^*=\sum_{j=1}^m\langle v_i^*, R(T(v_k))v_j\rangle v_j^*.$$
Then
\begin{eqnarray*}
&&\sum_{i,k=1}^m T(v_i)\otimes v_k^*\otimes L^*(T(v_k))v_i^*
=\sum_{i,k=1}^m \sum_{j=1}^m\langle v_j^*, L(T(v_k))v_i\rangle
T(v_j)\otimes v_k^*\otimes v_i^*\\&&=\sum_{i,k=1}^m T(\langle v_j^*,
L(T(v_k))v_i\rangle v_j)\otimes v_k^*\otimes v_i^*=\sum_{i,k=1}^m
T(L(T(v_k))v_i)\otimes v_k^*\otimes v_i^*.\end{eqnarray*} Hence, we
get\begin{eqnarray*} &&r_{12}\circ r_{13}+r_{13}\circ
r_{23}-r_{23}\circ r_{12}\\&&=\sum_{i,k=1}^m\{\{ T(v_i)\circ
T(v_k)-T(L(T(v_i))v_k)-T(R(T(v_k))v_i)\}\otimes v_i^*\otimes
v_k^*\\&&+v_k^*\otimes\{T(v_i)\circ
T(v_k)-T(L(T(v_i))v_k)-T(R(T(v_k)))v_i)\}\otimes
v_i^*\\&&+v_i^*\otimes v_k^*\otimes\{T(v_i)\circ
T(v_k)-T(L(T(v_i))v_k)-T(R(T(v_k))v_i)\}\} .\end{eqnarray*} So $r$
is a solution of the alternative Yang-Baxter equation in
$A\ltimes_{R^*,L^*}V^*$ if and only if
 $T$ is an $Al$-operator of $A$ associated to
 $(V,L,R)$.\end{proof}

\begin{prop} {\rm (\cite{G})}\quad Let $(A, \circ)$ be an alternative algebra
and $r\in A\otimes A$. Suppose that $r$ is skew-symmetric and
nondegenerate. Then $r$ is a solution of the alternative Yang-Baxter
equation in $A$ if and only if the bilinear form $\omega$ induced by
$r$ through equation (1.10) is a symplectic form.
\end{prop}

\begin{coro} Let $(A,\prec,\succ)$ be a pre-alternative algebra. Let $\{e_1,\cdots, e_n\}$
be a basis of $A$ and $\{e_1^*,\cdots, e_n^*\}$ be the dual basis.
Then
\begin{equation}
r=\sum_i{(e_i\otimes e_i^*-e_i^*\otimes e_i)},\end{equation} is a
nondegenerate solution of the alternative Yang-Baxter equation in
$As(A)\ltimes_{\mathfrak{r}^*_{\prec},\mathfrak{l}^*_{\succ}}A^*$.
Moreover, the symplectic form $\omega_p$ induced by $r$ through
equation (1.10) is given by
\begin{equation}\omega_p(x+a^*,y+b^*)=\langle a^*,y\rangle -\langle
x,b^*\rangle ,\;\; \forall x,y\in A,a^*,b^*\in A^*.\end{equation}
\end{coro}

\begin{proof} It follows from the fact that $T={\rm id}$ is an $Al$-operator of $As(A)$ associated to the bimodule
$(A,\mathfrak{l}_{\succ},\mathfrak{r}_{\prec})$.\end{proof}

\begin{prop}

Let ($A,\circ,\omega$) be an alternative algebra with a symplectic
form $\omega$. Suppose that there is a compatible pre-alternative
algebra structure ``$\prec,\succ$" on $A$ given by equation (2.16)
and a pre-alternative algebra structure ``$\prec_*, \succ_*$ on
$A^*$ given by equation (3.7), where the solution $r$ of the
alternative Yang-Baxter equation in $(A,\circ)$ is induced by
$\omega$ through equation (1.10). Let $a^*\ast b^*=a^* \prec_*
b^*+a^*\succ_* b^*$ for any $a^*,b^*\in A^*$. Then there is a
pre-alternative algebra structure ``$\prec_0,\succ_0$" on $A\oplus
A^*$ given by (for any $x,y\in A, a^*,b^*\in A^*$)
\begin{equation}
(x,a^*)\prec_0(y,b^*)=(x\prec
y+\mathfrak{l}_*^*(b^*)x,a^*\prec_*b^*+
\mathfrak{l}_{\circ}^*(y)a^*),
\end{equation}
\begin{equation}
(x,a^*)\succ_0(y,b^*)=(x\succ
y+\mathfrak{r}_*^*(a^*)y,a^*\succ_*b^*+
\mathfrak{r}_{\circ}^*(x)b^*).\end{equation} Moreover, its
associated alternative algebra is just the Drinfeld's double $D(A)$
for the coboundary alternative $D$-bialgebra $(A,\circ,\Delta_r)$

\end{prop}

\begin{proof} In fact, since $r$ is invertible, it is easy to show that (for
any $x,y\in A$ and $a^*,b^*\in A^*$)
$$\mathfrak{l}_*^*(b^*)x=x\prec T_r(b^*),
\mathfrak{l}_{\circ}^*(y)a^*=a\prec_*T_r^{-1}(y),
\mathfrak{r}_*^*(a^*)y=T_r(a^*)\succ y,
\mathfrak{r}_{\circ}^*(x)b^*=T_r^{-1}(x)\succ_*b^*.$$ So for any
$z\in A,c^*\in A^*$,

{\small \begin{eqnarray*} &
&((x,a^*)\succ_0(y,b^*))\prec_0(z,c^*)\\&=&((x+T_r(a^*))\succ
y,(T_r^{-1}(x)+a^*)\succ_*b^*)\prec_0(z,c^*)\\
&=&((x\succ y)\prec z+(x\succ y)\prec T_r(c^*)+(T_r(a^*)\succ
y)\prec z+(T_r(a^*)\succ y)\prec T_r(c^*),\\ & &(T_r^{-1}(x)\succ_*
b^*)\prec_*T_r^{-1}(z)+(T_r^{-1}(x)\succ_*b^*)\prec_*c^*+(a^*\succ
b^*)\prec_*T_r^{-1}(z)+(a^*\succ_*b^*)\prec_*c^*).
\end{eqnarray*}}
Similarly, {\small \begin{eqnarray*}
& &((y,b^*)\prec_0(x,a^*))\prec_0(z,c^*)\\
&=&((y\prec x)\prec z+(y\prec x)\prec T_r(c^*)+(y\prec
T_r(a^*))\prec z+(y\prec T_r(a^*))\prec T_r(c^*),\\
& &(b^*\prec_*T_r^{-1}(x))\prec_* T_r^{-1}(z)+(b^*\prec_*a^*)\prec_*
T_r^{-1}(z)+(b^*\prec_* T_r^{-1}(x))\prec_* c^*+(b^*\prec_*
a^*)\prec_*
c^*),\\
& &(x,a^*)\succ_0((y,b^*)\prec_0(z,c^*))\\
&=&(x\succ(y\prec T_r(c^*)+x\succ(y\prec z)+T_r(a^*)\succ(y\prec
z)+T_r(a^*)\succ(y\prec T_r(c^*)),\\
 &
&T_r^{-1}(x)\succ_*(b^*\prec_*
c^*)+T_r^{-1}(x)\succ_*(b^*\prec_*T_r^{-1}(z))+a^*\succ_*(b^*\prec_*c^*)+a^*\succ_*(b^*\prec_*T_r^{-1}(z)),\\
& &(y,b^*)\prec_*((x,a^*)\bullet(z,c^*))\\
&=&(y\prec(T_r(a^*)\circ T_r(c^*))+y\prec(T_r(a^*)\prec
z)+y\prec(x\succ T_r(c^*))+y\prec(x\circ z)+y\prec(x\prec
T_r(c^*))\\ & &+y\prec(T_r(a^*)\succ z), b^*\prec_*(a^*\ast
c^*)+b^*\prec_*(a^*\prec_*T_r^{-1}(z))+b^*\prec_*(T_r^{-1}(x)\succ_*c^*)+\\
& &b^*\prec_*(T_r^{-1}(x)\ast
T_r^{-1}(z))+b^*\prec_*(T_r^{-1}(x)\prec_*c^*)+b^*\prec_*(a^*\succ_*T_r^{-1}(z))),
\end{eqnarray*}}
where $\bullet=\prec_0+\succ_0$. Hence
\begin{eqnarray*}
&
&((x,a^*)\succ_0(y,b^*))\prec_0(z,c^*)+((y,b^*)\prec_0(x,a^*))\prec_0(z,c^*)\\
&=&
(x,a^*)\succ_0((y,b^*)\prec_0(z,c^*))+(y,b^*)\prec_*((x,a^*)\bullet(z,c^*)).
\end{eqnarray*}
By a similar study, we prove that  $(\prec_0,\succ_0)$ also
satisfies equation (2.8) and the second equation in equation (2.7).
\end{proof}

\begin{prop} Let $(A,\circ)$ be an alternative algebra.

(1) For any skew-symmetric solution $r$ of the alternative
Yang-Baxter equation in $(A,\circ)$, the Drinfeld's double $D(A)$
for the coboundary alternative $D$-bialgebra $(A,\circ,\Delta_r)$ is
isomorphic to
$A\ltimes_{\mathfrak{r}_{\circ}^*,\mathfrak{l}_{\circ}^*}A^*$ as
alternative algebras.

(2) The skew-symmetric
 solutions of the alternative Yang-Baxter equation in $(A,\circ)$ are in one to one
correspondence with the linear maps $T_r:A^*\rightarrow A$ whose
graphs
\begin{equation}
graph(T_r)=\{(T_r(a^*),a^*)\in
A\ltimes_{\mathfrak{r}_{\circ}^*,\mathfrak{l}_{\circ}^*}A^*|a^*\in
A^*\}\end{equation} are Lagrangian subalgebras of
$A\ltimes_{\mathfrak{r}_{\circ}^*,\mathfrak{l}_{\circ}^*}A^*$ with
respect to the bilinear form given by equation (3.2). Consequently
every alternative subalgebra  which is also a Lagrangian
$graph(T_r)$ of
$A\ltimes_{\mathfrak{r}_{\circ}^*,\mathfrak{l}_{\circ}^*}A^*$,
carries a pre-alternative algebra structure defined by (for any
$a^*,b^*\in A^*$)
\begin{equation}
(T_r(a^*),a^*)\prec(T_r(b^*),b^*)=
(T_r(\mathfrak{l}_{\circ}^*(T_r(b^*))a^*),\mathfrak{l}_{\circ}^*(T_r(b^*))a^*),\end{equation}
\begin{equation}
(T_r(a^*),a^*)\succ(T_r(b^*),b^*)=
(T_r(\mathfrak{r}_{\circ}^*(T_r(a^*))b^*),\mathfrak{r}_{\circ}^*(T_r(a^*))b^*).\end{equation}

\end{prop}

\begin{proof}
(1) Let $r$ be a skew-symmetric solution of  the alternative
Yang-Baxter equation in $(A,\circ)$. Let the product in $A^*(r)$ be
$\ast$, then by Proposition $3.6$ we know $a^*\ast
b^*=\mathfrak{r}_{\circ}^*(T_r(a^*))b^*+\mathfrak{l}_{\circ}^*(T_r(b^*)a^*)$,
for any $a^*,b^*\in A^*$. We claim that for any $x,y\in A,a^*,b^*\in
A^*$ we have
\begin{equation}
\mathfrak{r}_{\ast}^*(a^*)y+\mathfrak{l}_{\ast}^*(b^*)x=T_r(a^*)\circ
y+x\circ
T_r(b^*)-T_r(\mathfrak{r}_{\circ}^*(x)b^*+\mathfrak{l}_{\circ}^*(y)a^*).
\end{equation}
In fact, it follows from the following computations (for any $c^*\in
A^*$):
\begin{eqnarray*}
\langle\mathfrak{r}_*^*(a^*)y+\mathfrak{l}_*^*(b^*)x,c^*\rangle&=&\langle
y,c^*\ast
a^*\rangle+\langle x,b^*\ast c^*\rangle\\
&=&\langle
y,\mathfrak{r}_{\circ}^*(T_r(c^*))a^*+\mathfrak{l}_{\circ}^*(T_r(a^*))c^*\rangle+
\langle x,\mathfrak{r}_{\circ}^*(T_r(b^*))c^*+\mathfrak{l}_{\circ}^*(T_r(c^*))b^*\rangle\\
&=&\langle y\circ T_r(c^*),a^*\rangle+\langle T_r(a^*)\circ
y,c^*\rangle+\langle x\circ
T_r(b^*),c^*\rangle+\langle T_r(c^*)\circ x,b^*\rangle\\
&=&\langle T_r(a^*)\circ y+x\circ
T_r(b^*)-T_r(\mathfrak{r}_{\circ}^*(x)b^*+\mathfrak{l}_{\circ}^*(y)a^*),c^*\rangle,
\end{eqnarray*}
where we use the fact that $\langle T_r(a^*),b^*\rangle=-\langle
a^*,T_r(b^*)\rangle$, which follows from the fact that $r$ is
skew-symmetric. Define a linear map $\lambda:(D(A)=A\oplus
A^*,\bullet)\rightarrow
(A\ltimes_{\mathfrak{r}_{\circ}^*,\mathfrak{l}_{\circ}^*}A^*,\star)$
by $$\lambda((x,a^*))=(T_r(a^*)+x,a^*),\;\;\forall x\in A, a^*\in
A^*.$$ Then we have
\begin{eqnarray*}
\lambda((x,a^*))\star\lambda((y,b^*))
&=&((T_r(a^*)+x)\circ(T_r(b^*)+y),\mathfrak{r}^*_{\circ}(T_r(a^*)+x)b^*+
\mathfrak{l}_{\circ}^*(T_r(b^*)+x)a^*)\\
&=&(T_r(a^*\ast
b^*+\mathfrak{r}_{\circ}^*(x)b^*+\mathfrak{l}_{\circ}^*(y)a^*)+x\circ
y+\mathfrak{r}_{\ast}^*(a^*)y+\mathfrak{l}_*^*(b^*)x,a^*\ast
b^*+\\ & &\mathfrak{r}_{\circ}^*(x)b^*+\mathfrak{l}_{\circ}^*(y)a^*)\\
&=&\lambda((x,a^*)\bullet(y,b^*)),
\end{eqnarray*}
where equations (3.6) and (3.17) are used. Furthermore, it is easy
to see that $\lambda$ is bijective. Therefore $\lambda$ is an
isomorphism of alternative algebras.

(2) First, we have $\lambda(A^*(r))=graph(T_r)$. So $graph(T_r)$ is
a subalgebra of
$A\ltimes_{\mathfrak{r}_{\circ}^*,\mathfrak{l}_{\circ}^*}A^*$. Since
$r$ is skew-symmetric, $graph(T_r)$ is isotropic with respect to the
bilinear form defined by equation (3.2). Moreover, it has a
complementary isotropic algebra $\lambda(A)=A$. So it is a
Lagrangian subalgebra of
$A\ltimes_{\mathfrak{r}_{\circ}^*,\mathfrak{l}_{\circ}^*}A^*$.
Conversely, let $T:A^*\rightarrow A$ be a linear map whose
$graph(T)$ is a Lagrangian subalgebra of
$A\ltimes_{\mathfrak{r}_{\circ}^*,\mathfrak{l}_{\circ}^*}A^*$. So
$T$ is skew-symmetric, that is, $\langle T(a^*),b^*\rangle =-\langle
T(b^*),a^*\rangle $ for any $a^*,b^*\in A^*$. Since $graph(T)$ is a
subalgebra, we have that
\begin{eqnarray*}
(T(a^*),a^*)\star(T(b^*),b^*)&=&(T(a^*)\circ
T(b^*),\mathfrak{r}_{\circ}^*(T_r(a^*))b^*+\mathfrak{l}_{\circ}^*(T_r(b^*))a^*)\\
&=&(T_r(\mathfrak{r}_{\circ}^*(T_r(a^*))b^*+\mathfrak{l}_{\circ}^*(T_r(b^*))a^*),
\mathfrak{r}_{\circ}^*(T_r(a^*))b^*+\mathfrak{l}_{\circ}^*(T_r(b^*))a^*).\end{eqnarray*}
So $T(a^*)\circ
T(b^*)=T(\mathfrak{r}_{\circ}^*(T_r(a^*))b^*+\mathfrak{l}_{\circ}^*(T_r(b^*))a^*)$.
By Proposition 3.6, $T$ corresponds to certain skew-symmetric
solution of the alternative Yang-Baxter equation in $(A,\circ)$. The
last statement is obtained by transferring (by the isomorphism
$\lambda$) the pre-alternative algebra structure of $A^*(r)$ to
$graph(T_r)$.
\end{proof}

\section{Phase spaces of alternative algebras and matched pairs of
alternative algebras}

\setcounter{equation}{0}
\renewcommand{\theequation}
{4.\arabic{equation}}

\begin{defn}{\rm Let $(A,\circ,\omega)$ be a symplectic alternative algebra.
$A$ is called an {\it L-symplectic} alternative algebra if $A$ is a
direct sum of the underlying vector space of two Lagrangian
subalgebras $A^+$ and $A^{-}$. It is denoted by
$(A,\circ,A^+,A^-,\omega)$. Two $L$-symplectic alternative algebras
$(A_1,\circ,A_1^+,A_1^-,\omega_1)$ and
$(A_2,\circ,A_2^+,A_2^-,\omega_2)$ are {\it isomorphic} if there
exists an isomorphism  $\varphi:A_1\rightarrow A_2$ of alternative
algebras such that
\begin{equation}
\varphi(A_1^+)=A_2^+,\quad \varphi(A_1^-)=A_2^-,\quad
\omega_1(a,b)=\varphi^*\omega_2(a,b)=\omega_2(\varphi(a),\varphi(b)),\;\;
\forall a,b\in A_1.\end{equation} }\end{defn}

 It is easy to see that a symplectic
alternative algebra $(A,\circ,\omega)$ is an $L$-symplectic
alternative algebra if and only if $A$ is a direct sum of the
underlying vector space of two isotropic subalgebras.

\begin{prop} Let $(A,\circ,A^+,A^-,\omega)$
be an $L$-symplectic alternative algebra. Then there exists a
pre-alternative algebra structure on $A$ given by equation $(2.18)$
such that $A^+$ and $A^-$ are pre-alternative subalgebras. Moreover,
two $L$-symplectic alternative algebras
$(A_i,\circ,A_i^+,A_i^-,\omega_i)\;(i=1,2)$ are isomorphic if and
only if there exists an isomorphism of pre-alternative algebras
satisfying equation $(4.1)$ which the compatible pre-alternative
algebras are given by equation $(2.18)$.
\end{prop}

\begin{proof} If $a,b,c\in A^+$, then $\omega(a\prec b,c)=\omega(a,b\circ
c)=0$. Since $A^+$ is a Lagrangian subalgebra of $A$, we have
$a\prec b\in A^+,\forall a,b\in A^+$. Similar arguments apply to
``$\succ$" and $A^-$. So the first conclusion holds. It is
straightforward to get the second conclusion.\end{proof}

\begin{defn} {\rm Let $(A,\circ)$ be an alternative algebra. If there is an
alternative algebra structure on the direct sum of the underlying
vector space of $A$ and  $A^*$ such that $A$ and $A^*$ are
alternative subalgebras and the natural skew-symmetric bilinear form
$\omega_p$ on $A\oplus A^*$ given by equation (3.11) is a symplectic
form, then it is called a {\it phase space} of the alternative
algebra $A$.}\end{defn}

\begin{remark}{\rm  The notion of
phase space is borrowed from mathematical physics (\cite{K1},
\cite{Bai1}).}
\end{remark}

\begin{prop}
 Every $L$-symplectic alternative algebra
$(A,\circ,A^+,A^-,\omega)$ is isomorphic to a phase space of $A^+$.
\end{prop}

\begin{proof}
Since $A^{-}$ and $(A^+)^*$ are identified by the symplectic form,
we can transfer the alternative algebra structure on $A^{-}$ to
$(A^+)^*$. Hence the alternative algebra structure on $A^{+}\oplus
A^{-}$ can be transferred to $A^{+}\oplus(A^{+})^*$. Therefore the
conclusion follows.
\end{proof}

\begin{remark}{\rm By symmetry, every $L$-symplectic
alternative algebra $(A,\circ,A^+,A^-,\omega)$ is isomorphic to a
phase space of $A^-$.}
\end{remark}

\begin{prop} Let $(A,\circ)$ and $(B,\ast)$ be two alternative algebras.
Suppose that there are linear maps $L_A,R_A:A\rightarrow gl(B)$ and
$L_B,R_B:B\rightarrow gl(B)$ such that $(L_A,R_A)$ is a bimodule of
$A$ and $(L_B,R_B)$ is a bimodule of $B$ and they satisfy the
following conditions:
\begin{equation}
L_B(Ass_A(x)a)y+(Ass_B(a)x)\circ y=L_B(a)(x\circ
y)+R_B(R_A(y)a)x+x\circ(L_B(a)y),\end{equation}\begin{equation}
R_B(a)(x\circ y+y\circ
x)=R_B(L_A(y)a)x+x\circ(R_B(a)y)+R_B(L_A(x)a)y+y\circ(R_B(a)x),\end{equation}
\begin{equation}R_B(a)(x\circ
y)+L_B(L_A(x)a)y+(R_B(a)x)\circ
y=R_B(Ass_A(y)a)x+x\circ(Ass_B(a)y),\end{equation}
\begin{equation}L_B(a)(x\circ
y+y\circ x)=(L_B(a)x)\circ y+L_B(R_A(x)a)y+ (L_B(a)y)\circ
x+L_B(R_A(y)a)x,\end{equation}\begin{equation}L_A(Ass_B(a)x)b+(Ass_A(x)a)\ast
b=L_A(x)(a\ast
b)+R_A(R_B(b)x)a+a\ast(L_A(x)b),\end{equation}\begin{equation}R_A(x)(a\ast
b+b\ast
a)=R_A(L_B(b)x)a+a\ast(R_A(x)b)+R_A(L_B(a)x)b+b\ast(R_A(x)a),\end{equation}\begin{equation}R_A(x)(a\ast
b)+L_A(L_B(a)x)b+(R_A(x)a)\ast
b=R_A(Ass_B(b)x)a+a\ast(Ass_A(x)b),\end{equation}\begin{equation}L_A(x)(a\ast
b+b\ast a)=(L_A(x)a)\ast b+L_A(R_B(a)x)b+ (L_A(x)b)\ast
a+L_A(R_B(b)x)a,\end{equation} where $x,y\in A,a,b\in B$,
$Ass_i=L_i+R_i,i=A,B$. Then there is an alternative algebra
structure on the vector space $A\oplus B$ given by
\begin{equation}(x+a)\star(y+b)=(x\circ y+L_B(a)y+R_B(b)x)+(a\ast
b+L_A(x)b+R_A(y)a),\quad \forall x,y\in A,a,b\in B.\end{equation} We
denote this alternative algebra by $A\bowtie_{L_A,R_A}^{L_B,R_B}B$
or simply $A\bowtie B$. Moreover, $(A,B,L_A,R_A$, $L_B,R_B)$
satisfying the above conditions is called a $matched$ $pair$ of
alternative algebras. On the other hand, every alternative algebra
which is a direct sum of the underlying vector spaces of two
subalgebras can be obtained from the above way.\end{prop}

\begin{proof} Straightforward.\end{proof}

\begin{prop}
 Let $(A,\prec_1,\succ_1)$ be a
pre-alternative algebra and $(As(A),\circ)$ be the associated
alternative algebra. Suppose that there exists a pre-alternative
algebra structure ``$\prec_2,\succ_2$" on its dual space $A^*$ and
$(As(A^*),\ast)$ is the associated alternative algebra. Then there
exists an $L$-symplectic alternative algebra structure on the vector
space $A\oplus A^*$ such that $(As(A),\circ)$ and $(As(A^*),\ast)$
are Lagrangian subalgebras associated to the symplectic form
$(3.11)$ if and only if
$(As(A),As(A^*),\mathfrak{r}_{\prec_1}^*,\mathfrak{l}_{\succ_1}^*,
\mathfrak{r}_{\prec_2}^*,\mathfrak{l}_{\succ_2}^*)$ is a matched
pair of alternative algebras. Furthermore, every $L$-symplectic
alternative algebra can be obtained from the above way.
\end{prop}

\begin{proof}
If
$(As(A),As(A^*),\mathfrak{r}_{\prec_1}^*,\mathfrak{l}_{\succ_1}^*,
\mathfrak{r}_{\prec_2}^*,\mathfrak{l}_{\succ_2}^*)$ is a matched
pair of alternative algebras, then it is straightforward to show
that the bilinear form (3.11) is a symplectic form of the
alternative algebra
$As(A)\bowtie_{\mathfrak{r}_{\prec_1}^*,\mathfrak{l}_{\succ_1}^*}^{\mathfrak{r}_{\prec_2}^*,\mathfrak{l}_{\succ_2}^*}
As(A^*)$. Conversely, set
$$x\star a^*=L_{\circ}(x)a^*+R_{\ast}(a^*)x,\;\;a^*\star x=L_{\ast}(a^*)x+R_{\circ}(x)a^*,\;\;\forall\;
x\in A, a^*\in A^*,$$ where $\star$ is the alternative algebra
structure of
$As(A)\bowtie_{\mathfrak{r}_{\prec_1}^*,\mathfrak{l}_{\succ_1}^*}^{\mathfrak{r}_{\prec_2}^*,\mathfrak{l}_{\succ_2}^*}
As(A^*)$. Then $(A,A^*,L_{\circ},R_{\circ}$, $L_{\ast},R_{\ast})$ is
a matched pair of alternative algebras. Note that
\begin{eqnarray*}
\langle R_{\circ}(x) a^*, y\rangle&=& \langle a^*\star x,y\rangle
=-\omega_p(y,a^*\star x)=-\omega_p(x\succ_1 y,a^*)=\langle
\mathfrak{l}_{\succ_1}^*(x)a^*,y\rangle, \\
\langle L_{\ast}(a^*) x, b^*\rangle&=&\langle a^*\star x,b^*\rangle
=\omega_p(b^*, a^*\star x)=\omega_p(b^*\prec_2a^*, x)=\langle
\mathfrak{r}_{\prec_2}^*(a^*)x, b^*\rangle
\end{eqnarray*}
 where
$x,y\in A, a^*,b^*\in A^*$. Hence,
$R_{\circ}=\mathfrak{l}_{\succ_1}^*,L_{\ast}=\mathfrak{r}_{\prec_2}^*$.
Similarly,
$L_{\circ}=\mathfrak{r}_{\prec_1}^*,R_{\ast}=\mathfrak{l}_{\succ_2}^*$.
\end{proof}

\section{Bimodules and matched pairs of pre-alternative algebras}
\setcounter{equation}{0}
\renewcommand{\theequation}
{5.\arabic{equation}}

\begin{defn}{\rm
Let $(A,\prec,\succ)$ be a pre-alternative algebra and $V$ be a
vector space. Let
$L_{\prec},R_{\prec},L_{\succ},R_{\succ}:A\rightarrow gl(V)$ be four
linear maps. $V$ (or $(L_{\prec},R_{\prec},L_{\succ},R_{\succ})$, or
$(V,L_{\prec},R_{\prec},L_{\succ}$, $R_{\succ})$) is called a {\it
representation } or a {\it bimodule} of $A$ if (for any $x,y\in A$)
\begin{equation}L_{\succ}(x\circ y+y\circ
x)=L_{\succ}(x)L_{\succ}(y)+L_{\succ}(y)L_{\succ}(x),
\end{equation}
\begin{equation}
R_{\succ}(y)(L_{\circ}(x)+R_{\circ}(x))=L_{\succ}(x)R_{\succ}(y)+R_{\succ}(x\succ
y),\end{equation}
\begin{equation}
R_{\prec}(y)L_{\succ}(x)+R_{\prec}(y)R_{\prec}(x)=L_{\succ}(x)R_{\prec}(y)+R_{\prec}(x\circ
y),
\end{equation}
\begin{equation}
R_{\prec}(y)R_{\succ}(x)+R_{\prec}(y)L_{\prec}(x)=R_{\succ}(x\prec
y)+L_{\prec}(x)R_{\circ}(y),\end{equation}
\begin{equation}L_{\prec}(x\succ y)+L_{\prec}(y\prec
x)=L_{\succ}(x)L_{\prec}(y)+L_{\prec}(y)L_{\circ}(x),
\end{equation}
\begin{equation}
L_{\succ}(y\circ
x)+R_{\prec}(x)L_{\succ}(y)=L_{\succ}(y)L_{\succ}(x)+L_{\succ}(y)R_{\prec}(x),
\end{equation}
\begin{equation}
R_{\succ}(y)R_{\circ}(x)+R_{\prec}(x)R_{\succ}(y)=R_{\succ}(x\succ
y)+R_{\succ}(y\prec x),
\end{equation}
\begin{equation}
R_{\succ}(x)L_{\circ}(y)+L_{\prec}(y\succ
x)=L_{\succ}(y)R_{\succ}(x)+L_{\succ}(y)L_{\prec}(x),\end{equation}
\begin{equation}R_{\prec}(y)R_{\prec}(x)+R_{\prec}(x)R_{\prec}(y)=R_{\prec}(x\circ
y+y\circ x),\end{equation}
\begin{equation}R_{\prec}(y)L_{\prec}(x)+L_{\prec}(x\prec
y)=L_{\prec}(x)(R_{\circ}(y)+L_{\circ}(y)),\end{equation} where
$x\circ y=x\prec y+x\succ
y,L_{\circ}=L_{\succ}+L_{\prec},R_{\circ}=R_{\succ}+R_{\prec}$.}
\end{defn}

According to \cite{S3}, $(V,
L_{\prec},R_{\prec},L_{\succ},R_{\succ})$ is a bimodule of a
pre-alternative algebra $(A,\prec,\succ)$  if and only if the direct
sum $A\oplus V$ of vector spaces is turn into a pre-alternative
algebra (the {\it semidirect sum}) by defining multiplications in
$A\oplus V$ by (for any $x,y\in A,a,b\in V$)
\begin{equation}
(x+a)\prec(y+b)=x\prec y+L_{\prec}(x)b+R_{\prec}(y)a,\;\;
(x+a)\succ(y+b)=x\succ y+L_{\succ}(x)b+R_{\succ}(y)a. \end{equation}
 We denote it by
$A\ltimes_{L_{\prec},R_{\prec},L_{\succ},R_{\succ}}V$ or simply
$A\ltimes V$.

\begin{prop}
 Let
$(V,L_{\prec},R_{\prec},L_{\succ},R_{\succ})$ be a bimodule of a
pre-alternative algebra $(A,\prec,\succ)$. Let $(As(A),\circ)$ be
the associated alternative algebra.

{\rm (1)} Both $(V,L_{\succ},R_{\prec})$ and
$(V,L_{\circ}=L_{\prec}+L_{\succ},R_{\circ}=R_{\prec}+R_{\succ})$
are bimodules of $(As(A),\circ)$.

{\rm (2)} For any bimodule $(V,L,R)$ of $(As(A),\circ)$,
$(V,0,R,L,0)$ is a
bimodule of $(A,\prec,\succ)$.

{\rm (3)} $(V^*,-R_{\succ}^*,L_{\circ}^*,R_{\circ}^*,-L_{\prec}^*)$
 is a bimodule of $(A,\prec,\succ)$.
\end{prop}

\begin{proof}
We only give the proof of (3) as an example. The others are more or
less straightforward. In fact, since $(V,L_{\circ},R_{\circ})$ is a
bimodule of $As(A)$, $R_{\circ}(x^2)=R_{\circ}(x)R_{\circ}(x)$  for
any $x\in A$. Hence
$$R_{\circ}(x\circ y+y\circ
x)=R_{\circ}(x)R_{\circ}(y)+R_{\circ}(y)R_{\circ}(x),\;\;\forall
x,y\in A.$$ So for any $v\in V,u^*\in V^*$,
\begin{eqnarray*}
\langle R_{\circ}^*(x\circ y+y\circ x)u^*,v\rangle&=&\langle
u^*,R_{\circ}(x\circ y+y\circ x)v\rangle=\langle u^*,(R_{\circ}(x)R_{\circ}(y)
+R_{\circ}(y)R_{\circ}(x))v\rangle\\
&=&\langle
(R_{\circ}^*(x)R_{\circ}^*(y)+R_{\circ}^*(y)R_{\circ}^*(x))u^*,v\rangle.
\end{eqnarray*}
Therefore $R_{\circ}^*(x\circ y+y\circ
x)=R_{\circ}^*(x)R_{\circ}^*(y)+R_{\circ}^*(y)R_{\circ}^*(x)$.
Similarly, $(-R_{\succ}^*,L_{\circ}^*,R_{\circ}^*,-L_{\prec}^*)$
also satisfies equations (5.2)-(5.10).
\end{proof}

\begin{exam}{\rm
 Let
$(A,\prec,\succ)$ be a pre-alternative algebra. Then
$(\mathfrak{l}_{\prec},\mathfrak{r}_{\prec},\mathfrak{l}_{\succ},\mathfrak{r}_{\succ}),
(0,\mathfrak{r}_{\prec},\mathfrak{l}_{\succ},0)$,
$(0,\mathfrak{r}_{\circ},\mathfrak{l}_{\circ},0),
(0,\mathfrak{l}_{\circ}^*,\mathfrak{r}_{\circ}^*,0),
(0,\mathfrak{l}_{\succ}^*, \mathfrak{r}_{\prec}^*,0)$ and
$(-\mathfrak{r}_{\succ}^*,\mathfrak{l}_{\circ}^*,
\mathfrak{r}_{\circ}^*,-\mathfrak{l}_{\prec}^*)$ are  bimodules of
$(A,\prec,\succ)$.}
\end{exam}

\begin{defn} {\rm  Let $(A,\prec_A,\succ_A)$ and
$(B,\prec_B,\succ_B)$ be two pre-alternative algebras. Suppose that
there are linear maps
$L_{\prec_A},R_{\prec_A},L_{\succ_A},R_{\succ_A}:A\rightarrow gl(B)$
and $L_{\prec_B},R_{\prec_B},L_{\succ_B},R_{\succ_B}:B\rightarrow
gl(A)$ such that the following products on the vector space $A\oplus
B$ (for any $x,y\in A, a,b\in B$)
\begin{equation}
(x+a)\prec(y+b)=x\prec_Ay+L_{\prec_B}(a)y+R_{\prec_B}(b)x+
a\prec_Bb+L_{\prec_A}(x)b+R_{\prec_A}(y)a,\end{equation}
\begin{equation}(x+a)\succ(y+b)=x\succ_Ay+L_{\succ_B}(a)y+R_{\succ_B}(b)x+
a\succ_Bb+L_{\succ_A}(x)b+R_{\succ_A}(y)a,\end{equation} define a
pre-alternative algebra structure. Then
$(A,B,L_{\prec_A},R_{\prec_A},L_{\succ_A},R_{\succ_A},
L_{\prec_B},R_{\prec_B}$, $L_{\succ_B}$, $R_{\succ_B})$ is called
{\it a matched pair of pre-alternative algebras}, and we denote this
pre-alternative algebra by
$A\bowtie_{L_{\prec_A},R_{\prec_A},L_{\succ_A},R_{\succ_A}}
^{L_{\prec_B},R_{\prec_B},L_{\succ_B},R_{\succ_B}}B$ or simply
$A\bowtie B$.}
\end{defn}

\begin{remark}{\rm It is not hard to write down explicitly a necessary and sufficient condition
satisfied by the above linear maps in a matched pair of
pre-alternative like Proposition 4.7, which involves 20 equations.
Since we will not use such a conclusion directly in this paper, they
are omitted here. Note that
$(B,L_{\prec_A},R_{\prec_A},L_{\succ_A},R_{\succ_A})$ and
$(A,L_{\prec_B},R_{\prec_B},L_{\succ_B},R_{\succ_B})$ must be
bimodules of $A$ and $B$, respectively.}\end{remark}

\begin{coro} Let $(A,B,L_{\prec_A},R_{\prec_A},L_{\succ_A},R_{\succ_A},
L_{\prec_B},R_{\prec_B},L_{\succ_B},R_{\succ_B})$ be a matched pair
of pre-alternative algebras. Then
$(As(A),As(B),L_{\prec_A}+L_{\succ_A},R_{\prec_A}+R_{\succ_A},
L_{\prec_B}+L_{\succ_B},R_{\prec_B}+R_{\succ_B})$ is a matched pair
of alternative algebras.
\end{coro}

\begin{proof}
It follows from the relation between the pre-alternative algebra and
the associated alternative algebra.
\end{proof}

\begin{prop}
 Let $(A,\prec_1,\succ_1)$ be a
pre-alternative algebra and $(As(A),\circ_1)$ be the associated
alternative algebra. Suppose there is a pre-alternative algebra
structure $``\prec_2,\succ_2$'' on its dual space $A^*$ and
$(As(A^*),\circ_2)$ is the associated alternative algebra. Then
$(As(A),As(A^*), \mathfrak{r}_{\prec_1}^*,\mathfrak{l}_{\succ_1}^*$,
$\mathfrak{r}_{\prec_2}^*, \mathfrak{l}_{\succ_2}^*)$ is a matched
pair of alternative algebras if and only if
$(A,A^*,-\mathfrak{r}_{\succ_1}^*,\mathfrak{l}_{\circ_1}^*,\mathfrak{r}_{\circ_1}^*$,
$-\mathfrak{l}_{\prec_1}^*$,
$-\mathfrak{r}_{\succ_2}^*,\mathfrak{l}_{\circ_2}^*,
\mathfrak{r}_{\circ_2}^*,-\mathfrak{l}_{\prec_2}^*)$ is a matched
pair of pre-alternative algebras.
\end{prop}

\begin{proof} By Corollary 5.6, we only need to prove the``only if" part of the conclusion.
 If $(As(A)$, $As(A^*),
\mathfrak{r}_{\prec_1}^*,\mathfrak{l}_{\succ_1}^*,\mathfrak{r}_{\prec_2}^*$,
$\mathfrak{l}_{\succ_2}^*)$ is a matched pair of alternative
algebras, then by Proposition 4.8,
$As(A)\bowtie_{\mathfrak{r}_{\prec_1}^*,
\mathfrak{l}_{\succ_1}^*}^{\mathfrak{r}_{\prec_2}^*,\mathfrak{l}_{\succ_2}^*}As(B)$
is an $L$-symplectic alternative algebra with the symplectic form
given by equation (3.11). Hence there is a compatible
pre-alternative algebra structure on
$As(A)\bowtie_{\mathfrak{r}_{\prec_1}^*,
\mathfrak{l}_{\succ_1}^*}^{\mathfrak{r}_{\prec_2}^*,\mathfrak{l}_{\succ_2}^*}As(B)$
given by equation (2.18). Then for any $x,y\in A,a^*,b^*\in A^*$ we
have
$$
\langle a^*\prec x,y\rangle =\omega_p(a^*\prec
x,y)=\omega_p(a^*,x\circ_1 y)=\langle
\mathfrak{l}_{\circ_1}^*(x)a^*,y\rangle,
$$
$$\langle a^*\prec x,b^*\rangle
=-\omega_p(a^*\prec x,b^*)=-\langle a^*,
\mathfrak{l}_{\succ_2}^*(b^*) x\rangle=-\langle b^*\succ_2
a^*,x\rangle=\langle -\mathfrak{r}_{\succ_2}^*(a^*)x, b^*\rangle.$$
So $a^*\prec
x=-\mathfrak{r}_{\succ_2}^*(a^*)x+\mathfrak{l}_{\circ_1}(x)a^*$.
Similarly, $$x\prec
a^*=-\mathfrak{r}_{\succ_1}^*(x)a^*+\mathfrak{l}_{\circ_2}(a^*)x,
x\succ
a^*=\mathfrak{r}_{\circ_1}^*(x)a^*-\mathfrak{l}_{\prec_2}^*(a^*)x,
a^*\succ
x=\mathfrak{r}_{\circ_2}(a^*)x-\mathfrak{l}_{\prec_1}^*(x)a^*.$$
Therefore
$(A,A^*,-\mathfrak{r}_{\succ_1}^*,\mathfrak{l}_{\circ_1}^*,\mathfrak{r}_{\circ_1}^*,
-\mathfrak{l}_{\prec_1}^*,-\mathfrak{r}_{\succ_2}^*,\mathfrak{l}_{\circ_2}^*,
\mathfrak{r}_{\circ_2}^*,-\mathfrak{l}_{\prec_2}^*)$ is a matched
pair of pre-alternative algebras.\end{proof}

\section{Pre-alternative bialgebras}
\setcounter{equation}{0}
\renewcommand{\theequation}
{6.\arabic{equation}}

\begin{theorem}
 Let $(A,\prec,\succ,\alpha,\beta)$ be a
pre-alternative algebra $(A,\prec,\succ)$ equipped with two
comultiplications $\alpha,\beta:A\rightarrow A\otimes A$ and
$(As(A),\circ)$ be the associated alternative algebra. Suppose that
$\alpha^*,\beta^*:A^*\otimes A^*\rightarrow A^*$ induce a
pre-alternative algebra structure ``$\prec_*,\succ_*$'' on its dual
space $A^*$. Then
$(As(A),As(A^*),\mathfrak{r}_{\prec}^*,\mathfrak{l}_{\succ}^*,
\mathfrak{r}_{\prec_*}^*,\mathfrak{l}_{\succ_*}^*)$ is a matched
pair of alternative algebras if and only if $\alpha,\beta$ satisfy
the following eight equations (for any $x,y\in A$):

{\small
\begin{equation} \alpha(x\circ y+y\circ
x)=(\mathfrak{r}_{\circ}(y)\otimes 1+1\otimes
\mathfrak{l}_{\succ}(y))\alpha(x)+(\mathfrak{r}_{\circ}(x)\otimes
1+1\otimes \mathfrak{l}_{\succ}(x))\alpha(y),
\end{equation}
\begin{equation}
\beta(x\circ y+y\circ x)=(\mathfrak{r}_{\prec}(y)\otimes 1+1\otimes
\mathfrak{l}_{\circ}(y))\beta(x)+(\mathfrak{r}_{\prec}(x)\otimes
1+1\otimes \mathfrak{l}_{\circ}(x))\beta(y),\end{equation}
\begin{equation}\alpha(x\circ y)=(1\otimes \mathfrak{r}_{\prec}(x)+1\otimes
\mathfrak{l}_{\succ}(x)-\mathfrak{l}_{\circ}(x)\otimes
1)\alpha(y)+(\mathfrak{r}_{\circ}(y)\otimes
1)\alpha(x)+(\mathfrak{r}_{\circ}(y)\otimes 1-1\otimes
\mathfrak{l}_{\succ}(y))\tau\beta(x),\end{equation}
\begin{equation}\beta(x\circ y)=(\mathfrak{l}_{\succ}(y)\otimes 1+\mathfrak{r}_{\prec}(y)\otimes
1-1\otimes \mathfrak{r}_{\circ}(y))\beta(x)+(1\otimes
\mathfrak{l}_{\circ}(x))\beta(y)+(1\otimes
\mathfrak{l}_{\circ}(x)-\mathfrak{r}_{\prec}(x)\otimes
1)\tau\alpha(y),\end{equation}
\begin{equation}(\alpha+\beta)(x\prec y)=(1\otimes
\mathfrak{l}_{\prec}(x))(\tau\alpha+\beta)(y)+(\mathfrak{r}_{\prec}(y)\otimes
1+\mathfrak{l}_{\succ}(y)\otimes 1-1\otimes
\mathfrak{r}_{\prec}(y))(\alpha+\beta)(x)-(\mathfrak{r}_{\succ}(x)\otimes
1)\tau\beta(y),\end{equation}
\begin{equation}(\alpha+\beta)(x\succ y)=(\mathfrak{r}_{\succ}(y)\otimes
1)(\alpha+\tau\beta)(x)+(1\otimes \mathfrak{l}_{\succ}(x)+1\otimes
\mathfrak{r}_{\prec}(x)-\mathfrak{l}_{\succ}(x)\otimes
1)(\alpha+\beta)(y)-(1\otimes
\mathfrak{l}_{\prec}(y))\tau\alpha(x),\end{equation}
\begin{equation}(\alpha+\beta+\tau\alpha+\tau\beta)(x\succ y)=(\mathfrak{r}_{\succ}(y)\otimes
1)\alpha(x)+(1\otimes
\mathfrak{l}_{\succ}(x))(\alpha+\beta)(y)+(1\otimes
\mathfrak{r}_{\succ}(y))\tau\alpha(x)+(\mathfrak{l}_{\succ}(x)\otimes
1)(\tau\alpha+\tau\beta)(y),\end{equation}
\begin{equation}(\alpha+\beta+\tau\alpha+\tau\beta)(x\prec y)=(1\otimes
\mathfrak{l}_{\prec}(x))\beta(y)+(\mathfrak{r}_{\prec}(y)\otimes
1)(\alpha+\beta)(x)+(\mathfrak{l}_{\prec}(x)\otimes
1)\tau\beta(y)+(1\otimes
\mathfrak{r}_{\prec}(y))(\tau\alpha+\tau\beta)(x).\end{equation}}
\end{theorem}

\begin{proof}
 By Proposition 4.7, we just need to prove
equations (6.1)-(6.8) are equivalent to equations (4.2)-(4.9), where
$(A,B,L_A,R_A,L_B,R_B)$ is replaced by
$(As(A),As(A^*),\mathfrak{r}_{\prec}^*,\mathfrak{l}_{\succ}^*,
\mathfrak{r}_{\prec_*}^*,\mathfrak{l}_{\succ_*}^*)$. As an example,
we give an explicit proof of the equivalence between equations (4.2)
and (6.3). The proof of the others is similar. In fact, in this
case, equation (4.2) becomes
$$\mathfrak{r}_{\prec}^*(\mathfrak{r}_{\prec}^*(x)a^*+
\mathfrak{l}_{\succ}^*(x)a^*)y+(\mathfrak{r}_{\prec}^*(a^*)x+
\mathfrak{l}_{\succ}^*(a^*)x)\circ
y=\mathfrak{r}_{\prec}^*(a^*)(x\circ
y)+\mathfrak{l}_{\succ}^*(\mathfrak{l}_{\succ}^*(y)a^*)x
+x\circ(\mathfrak{r}_{\prec}^*(a^*)y),$$ where $x,y\in A, a^*\in
A^*$. Let both the left-hand side and the right-hand side of the
above equation act on $b^*\in A^*$, we have
\begin{eqnarray*}
\langle LHS,b^*\rangle &=&\langle
\mathfrak{r}_{\prec}^*(\mathfrak{r}_{\prec}^*(x)a^*+
\mathfrak{l}_{\succ}^*(x)a^*)y+(\mathfrak{r}_{\prec}^*(a^*)x+
\mathfrak{l}_{\succ}^*(a^*)x)\circ y,b^*\rangle \\
& =&\langle
y,b^*\prec(\mathfrak{r}_{\prec}^*(x)a^*+\mathfrak{l}_{\succ}^*(x)a^*)\rangle
+ \langle \mathfrak{r}_{\prec}^*(a^*)x+
\mathfrak{l}_{\succ}^*(a^*)x,\mathfrak{r}_{\circ}^*(y)b^*\rangle\\
&=&\langle \alpha(y),b^*\otimes
\mathfrak{r}_{\prec}^*(x)a^*+b^*\otimes
\mathfrak{l}_{\succ}^*(x)a^*\rangle +\langle
\alpha(x),\mathfrak{r}_{\circ}^*(y)b^*\otimes a^*\rangle +\langle
\beta(x),a^*\otimes \mathfrak{r}_{\circ}^*(y)b^*\rangle\\
& =&\langle (1\otimes \mathfrak{r}_{\prec}(x)+1\otimes
\mathfrak{l}_{\succ}(x))\alpha(y)+(\mathfrak{r}_{\circ}(y)\otimes
1)(\alpha+\tau\beta)(x),b^*\otimes a^*\rangle;\\
\langle RHS,b^*\rangle& =&\langle x\circ y,b^*\prec a^*\rangle
+\langle x,(\mathfrak{l}_{\succ}^*(y)a^*)\succ b^*\rangle +\langle
\mathfrak{r}_{\prec}^*(a^*)y,\mathfrak{l}_{\circ}^*(x)b^*\rangle\\
&=&\langle \alpha(x\circ y),b^*\otimes a^*\rangle +\langle
\beta(x),\mathfrak{l}_{\succ}^*(y)a^*\otimes b^*\rangle +\langle
\alpha(y),\mathfrak{l}_{\circ}^*(x)b^*\otimes a^*\rangle\\
&=&\langle \alpha(x\circ
y)+(1\otimes\mathfrak{l}_{\succ}(y))\tau\beta(x)+(\mathfrak{l}_{\circ}(x)\otimes
1)\alpha(y),b^*\otimes a^*\rangle .\end{eqnarray*}
 So equation (4.2) holds if and only if equation (6.3) holds.
\end{proof}

\begin{defn}
{\rm ${\rm (1)}$ Let $(A,\alpha,\beta)$ be a vector space with two
comultiplications $\alpha,\beta:A\rightarrow A\otimes A$. If
$(A,\alpha^*,\beta^*)$ becomes a pre-alternative algebra, then we
call the triple $(A,\alpha,\beta)$ a {\it pre-alternative coalgebra}.\\
 ${\rm (2)}$ If $(A,\prec,\succ,\alpha,\beta)$ is a pre-alternative algebra
$(A,\prec,\succ)$ with two comultiplications
$\alpha,\beta:A\rightarrow A\otimes A$ such that $(A,\alpha,\beta)$
is a pre-alternative coalgebra and $\alpha,\beta$ satisfy equations
(6.1)-(6.8),
 then $(A,\prec,\succ,\alpha,\beta)$ is called a {\it pre-alternative
 bialgebra}.}
\end{defn}

 Combining Propositions 4.8, 5.7 and Theorem 6.1, we have the following
  conclusion:

\begin{coro} Let $(A,\prec_1,\succ_1)$ be a pre-alternative algebra
and $(As(A),\circ_1)$ be the associated alternative algebra. Let
$\alpha,\beta:A\rightarrow A\otimes A$ be two linear maps such that
$\alpha^*,\beta^*:A^*\otimes A^*\subset (A\otimes A)^*\rightarrow
A^*$ induce a pre-alternative algebra structure on $A^*$ denoted by
``$\prec_2,\succ_2$", that is, $(A,\alpha,\beta)$ is a
pre-alternative coalgebra. Let $(As(A^*),\circ_2)$ be the associated
alternative algebra of $(A^*,\prec_2,\succ_2)$. Then the following
conditions are equivalent:

${\rm (1)}$ $(As(A)\bowtie As(A^*),As(A),As(A^*), \omega_p)$ is an
$L$-symplectic alternative algebra (or a phase space of $As(A)$),
where $\omega_p$ is given by equation $(3.11)$;

 ${\rm (2)}$
$(As(A),As(A^*),\mathfrak{r}_{\prec_1}^*,\mathfrak{l}_{\succ_1}^*,
\mathfrak{r}_{\prec_2}^*,\mathfrak{l}_{\succ_2}^*)$ is a matched
pair of alternative algebras;

 ${\rm (3)}$
$(A,A^*,-\mathfrak{r}_{\succ_1}^*,\mathfrak{l}_{\circ_1}^*,\mathfrak{r}_{\circ_1}^*,
       -\mathfrak{l}_{\prec_1}^*,-\mathfrak{r}_{\succ_2}^*,\mathfrak{l}_{\circ_2}^*,
       \mathfrak{r}_{\circ_2}^*,-\mathfrak{l}_{\prec_2}^*)$
       is a matched pair of pre-alternative algebras;

${\rm (4)}$ $(A,\prec_1,\succ_1,\alpha,\beta)$ is a pre-alternative
bialgebra.
\end{coro}

\begin{defn}
{\rm Let $(A,\prec_A,\succ_A,\alpha_A,\beta_A)$ and
$(B,\prec_B,\succ_B,\alpha_B,\beta_B)$ be two pre-alternative
bialgebras. A {\it homomorphism of pre-alternative bialgebras}
$\varphi:A\rightarrow B$ is a homomorphism of pre-alternative
algebras such that
\begin{equation}(\varphi\otimes\varphi)\alpha_A(x)=\alpha_B(\varphi(x)),\quad
(\varphi\otimes\varphi)\beta_A(x)=\beta_B(\varphi(x)),\quad \forall
x\in A.\end{equation} }
\end{defn}

\begin{prop}
 Two $L$-symplectic
(hence phase spaces of) alternative algebras are isomorphic if and
only if their corresponding pre-alternative bialgebras are
isomorphic.
\end{prop}

\begin{proof}
Let $(As(A)\bowtie As(A^*),As(A),As(A^*),\omega_p)$ and
$(As(B)\bowtie As(B^*),As(B),As(B^*),\omega_p)$ be two
$L$-symplectic alternative algebras and $\varphi:As(A)\bowtie
As(A^*)\rightarrow As(B)\bowtie As(B^*)$ is the isomorphism. Then
$\varphi|_A:A\rightarrow B$ and $\varphi|_{A^*}:A^*\rightarrow B^*$
are isomorphisms of pre-alternative algebras due to Proposition 4.2.
Moreover, $\varphi|_{A^*}={(\varphi|_A)^*}^{-1}$ since
\begin{eqnarray*}
\langle \varphi|_{A^*}(a^*),\varphi(x)\rangle &=&
\omega_p(\varphi|_{A^*}(a^*),\varphi(x))= \omega_p(a^*,x) = \langle
a^*,x\rangle =\langle \varphi^*{(\varphi|_A)^*}^{-1}(a^*),x\rangle\\
& =& \langle {(\varphi|_A)^*}^{-1}(a^*),\varphi(x)\rangle,\;\;
\forall x\in A,a^*\in A^*.
\end{eqnarray*}
So $(\varphi|_A)^*:B^*\rightarrow A^*$ is a homomorphism of
pre-alternative algebras and then
$(A,\prec_A,\succ_A,\alpha_A,\beta_A)$ and
$(B,\prec_B,\succ_B,\alpha_B,\beta_B)$ are isomorphic as
pre-alternative bialgebras.
 Conversely, let
$(A,\prec_A,\succ_A,\alpha_A,\beta_A)$ and
$(B,\prec_B,\succ_B,\alpha_B,\beta_B)$ be two isomorphic
pre-alternative bialgebras, and  $\varphi':A\rightarrow B$ be the
isomorphism. Let $\varphi:A\oplus A^*\rightarrow B\oplus B^*$ be a
linear map defined by
$$\varphi(x)=\varphi'(x),\;\;\varphi(a^*)=(\varphi'^*)^{-1}(a^*),\;\;\forall
x\in A,a^*\in A^*.$$ Then it is easy to show that $\varphi$ is an
isomorphism of the two $L$-symplectic alternative algebras
$(As(A)\bowtie As(A^*),As(A),As(A^*),\omega_p)$ and $(As(B)\bowtie
As(B^*),As(B),As(B^*),\omega_p)$.\end{proof}

\begin{exam}{\rm Let $(A,\prec,\succ,\alpha,\beta)$ be a pre-alternative bialgebra.
Then its dual $(A,\prec_*,\succ_*$, $\gamma,\delta)$ is also a
pre-alternative bialgebra, where the pre-alternative algebra
structure $\prec,\succ$ on $A$ is defined by the linear maps
$\gamma^*,\delta^*:A\otimes A\rightarrow A$ and
$\alpha^*,\beta^*:A^*\otimes A^*\rightarrow A^*$ induce a
pre-alternative algebra structure on $A^*$ denoted by
$``\prec_*,\succ_*"$. }
\end{exam}

\begin{exam}{\rm Let $(A,\prec,\succ)$ be a pre-alternative algebra.
Then $(A,\prec,\succ,\ 0, 0)$ is a pre-alternative  bialgebra, and
the corresponding pre-alternative algebra structure on $A\oplus A^*$
is the semi-direct sum
$A\ltimes_{-\mathfrak{r}_{\succ}^*,\mathfrak{l}_{\circ}^*,\mathfrak{r}_{\circ}^*,
       -\mathfrak{l}_{\prec}^*}A^*$. Moreover, its corresponding
associated   alternative algebra is the semi-direct sum
$As(A)\ltimes_{\mathfrak{r}_{\prec}^*,\mathfrak{l}_{\succ}^*}A^*$
 with the symplectic form $\omega_p$ given by equation (3.11).
}\end{exam}

\section{Coboundary pre-alternative bialgebras}
\setcounter{equation}{0}
\renewcommand{\theequation}
{7.\arabic{equation}}

\begin{defn} {\rm A pre-alternative bialgebra
$(A,\prec,\succ,\alpha,\beta)$ is called {\it coboundary} if the
linear maps $\alpha,\beta:A\rightarrow A\otimes A$ are given by
\begin{equation}
\alpha(x)=(\mathfrak{r}_{\circ}(x)\otimes
1-1\otimes\mathfrak{l}_{\succ}(x))r_{\prec},
 \end{equation}
\begin{equation}
\beta(x)=(1\otimes\mathfrak{l}_{\circ}(x)-\mathfrak{r}_{\prec}(x)\otimes
 1)r_{\succ},\end{equation}
 where  $x\circ y=x\prec y+x\succ y$, $x,y\in A$ and $r_{\prec},r_{\succ}\in A\otimes
 A$.}
\end{defn}

\begin{remark}{\rm
 The expression of equations (7.1)-(7.2) and equations (6.1)-(6.2)
 looks like certain kind of ``$1$-coboundary" and ``1-cocycle".}\end{remark}

\begin{theorem}
Let $(A,\prec,\succ)$ be a pre-alternative algebra with
 two linear maps $\alpha,\beta:A\rightarrow A\otimes A$ defined by
 equations (7.1) and (7.2) respectively. If $r_{\prec}=r_{\succ}=r\in A\otimes A$ and
 $r$ is symmetric, then $\alpha,\beta$ satisfy
 equations (6.1)-(6.8).
\end{theorem}

\begin{proof} It is obvious that $\alpha,\beta$ automatically satisfy equations (6.1) and (6.2).
For equations (6.3)-(6.8), as an example, we give an explicit proof
of the fact that $\alpha,\beta$ satisfy equation (6.5) since the
proof of the other cases are similar.
 Assume $r=\sum_i{u_i\otimes v_i}\in A\otimes A$. After rearranging the terms
 suitably, we have that (note that we use the fact that $r$ is
 symmetric)
\begin{eqnarray*}
&&(\alpha+\beta)(x\prec y)-
 (1\otimes\mathfrak{l}_{\prec}(x))(\tau\alpha+\beta)(y)-(\mathfrak{r}_{\prec}(y)\otimes
1+\mathfrak{l}_{\succ}(y)\otimes 1-1\otimes
\mathfrak{r}_{\prec}(y))(\alpha+\beta)(x)\\
&\mbox{}&+(\mathfrak{r}_{\succ}(x)\otimes
1)\tau\beta(y)\\
&&=\sum_i\{u_i\circ(x\prec y)\otimes v_i-u_i\prec(x\prec y)\otimes
v_i-(u_i\circ x)\prec y\otimes v_i+(u_i\prec x)\prec y\otimes
v_i\\
\end{eqnarray*}
\begin{eqnarray*}
& &-y\succ(u_i\circ x)\otimes v_i+y\succ(u_i\prec x)\otimes
v_i+(y\circ u_i)\succ x\otimes v_i-u_i\otimes(x\prec y)\succ
v_i\\
& &+u_i\otimes (x\prec y)\circ v_i-u_i\otimes x\prec(v_i\circ
y)-u_i\otimes x\prec(y\circ v_i)-u_i\otimes(x\succ v_i)\prec
y\\
& &+u_i\otimes(x\circ v_i)\prec y+y\succ u_i\otimes x\prec
v_i+y\succ u_i\otimes x\succ v_i-y\succ
u_i\otimes x\circ v_i+u_i\prec y\otimes x\prec v_i\\
& &+u_i\prec y\otimes x\succ v_i-u_i\prec y\otimes x\circ
v_i+u_i\circ x\otimes v_i\prec y-u_i\prec x\otimes v_i\prec
y-u_i\succ x\otimes v_i\prec y\}.
\end{eqnarray*} The sum of the first seven terms is zero since it equals to
$$\sum_i{[u_i\succ(x\prec y)-(u_i\succ x)\prec y-y\succ(u_i\succ
x)+(y\circ u_i)\succ x]\otimes v_i}=0.$$ The sum of the 8th-13th
terms is zero since it equals to
$$\sum_i{u_i\otimes[(x\prec y)\prec v_i-x\prec(y\circ
v_i)-x\prec(v_i\circ y)+(x\prec v_i)\prec y]}=0.$$  The sum of the
14th-16th terms, the sum of 17th-19th terms and the sum of the last
three terms are all zero obviously.
\end{proof}

\begin{lemma}
 Let $A$ be a vector space and
$\alpha,\beta:A\rightarrow A\otimes A$ be two linear maps. Then
$(A,\alpha,\beta)$ is a pre-alternative coalgebra if and only if the
linear maps $S^i_{\alpha,\beta}:A\rightarrow A\otimes A\otimes A$
$(i\in \{1,2,3,4\})$ given by the following equations are all zero:
(for any $x\in A$)
\begin{equation}
S^1_{\alpha,\beta}(x)=((\alpha+\beta)\otimes 1)\beta(x)+(\tau\otimes
1)((\alpha+\beta)\otimes
1)\beta(x)-(1\otimes\beta)\beta(x)-(\tau\otimes
1)(1\otimes\beta)\beta(x),\end{equation}
\begin{equation}
S^2_{\alpha,\beta}(x)=(\beta\otimes 1)\alpha(x)+(\tau\otimes
1)(\alpha\otimes 1)\alpha(x)-(1\otimes\alpha)\beta(x)-(\tau\otimes
1)(1\otimes(\alpha+\beta))\alpha(x),\end{equation}
\begin{equation}S^3_{\alpha,\beta}(x)=((\alpha+\beta)\otimes 1)\beta(x)+(1\otimes\tau)(\beta\otimes
1)\alpha(x)-(1\otimes\beta)\beta(x)-(1\otimes\tau)(1\otimes\alpha)\beta(x),\end{equation}
\begin{equation}S^4_{\alpha,\beta}(x)=(\alpha\otimes 1)\alpha(x)+(1\otimes\tau)(\alpha\otimes
1)\alpha(x)-(1\otimes(\alpha+\beta))\alpha(x)-
(1\otimes\tau)(1\otimes(\alpha+\beta))\alpha(x).\end{equation}
\end{lemma}

\begin{proof}
It follows immediately from the definition of pre-alternative
algebra (cf. Remark 2.6). \end{proof}

\begin{defn}
{\rm Let $(A,\prec,\succ)$ be a pre-alternative algebra and
$(As(A),\circ)$ be the associated alternative algebra. Let $r\in
A\otimes A$. The following equations are called {\it
$PA_j^i$-equations} $(i=1,2,j=1,2,3)$ respectively:
\begin{equation}
PA_1^1=r_{12}\circ r_{13}-r_{23}\succ r_{12}-r_{13}\prec
r_{23}=0,\end{equation}
\begin{equation}PA_1^2=r_{13}\circ r_{12}-r_{12}\prec r_{23}-r_{23}\succ r_{13}=0,\end{equation}
\begin{equation}PA_2^1=r_{12}\circ r_{23}-r_{23}\prec r_{13}-r_{13}\succ r_{12}=0,\end{equation}
\begin{equation}PA_2^2=r_{23}\circ r_{12}-r_{13}\succ r_{23}-r_{12}\prec r_{13}=0,\end{equation}
\begin{equation}PA_3^1=r_{13}\circ r_{23}-r_{12}\succ r_{13}-r_{23}\prec r_{12}=0,\end{equation}
\begin{equation}PA_3^2=r_{23}\circ r_{13}-r_{13}\prec r_{12}-r_{12}\succ r_{23}=0.\end{equation}
 We set $PA_j=PA_j^1+PA_j^2$,
where $j=1,2,3$. All $PA_j^i$-equations $(i=1,2,j=1,2,3)$ are called
{\it $PA$-equations}.}
\end{defn}

\begin{prop}
Let $(A,\prec,\succ)$ be a pre-alternative algebra and
$(As(A),\circ)$ be the associated alternative algebra. Let $r\in
A\otimes A$ be symmetric. Let $\alpha,\beta:A\rightarrow A\otimes A$
be two linear maps given by equations (7.1) and (7.2) respectively,
where $r_{\prec}=r_{\succ}=r$. Then $(A,\alpha,\beta)$ becomes a
pre-alternative coalgebra if and only if the following equations
holds: (for any $x\in A$)
\begin{equation}
-(1\otimes 1\otimes \mathfrak{l}_{\circ}(x))PA_3+
(1\otimes\mathfrak{r}_{\prec}(x)\otimes
1)PA_3^2+(\mathfrak{r}_{\prec}(x)\otimes 1\otimes
1)PA_3^1=0,\end{equation}
\begin{equation}-(1\otimes 1\otimes\mathfrak{l}_{\succ}(x))PA_2+
(1\otimes\mathfrak{r}_{\circ}(x)\otimes
1)PA_2^1+(\mathfrak{r}_{\prec}(x)\otimes 1\otimes
1)PA_2^2=0,\end{equation}
\begin{equation}(\mathfrak{r}_{\prec}(x)\otimes 1\otimes 1)PA_3-(1\otimes
1\otimes\mathfrak{l}_{\circ}(x))PA_3^1-(1\otimes\mathfrak{l}_{\succ}(x)\otimes
1)PA_3^2=0,\end{equation}
\begin{equation}(\mathfrak{r}_{\circ}(x)\otimes 1\otimes
1)PA_1-(1\otimes\mathfrak{l}_{\succ}(x)\otimes 1)PA_1^2-(1\otimes
1\otimes\mathfrak{l}_{\succ}(x))PA_1^1=0.\end{equation}
\end{prop}

\begin{proof} We give an explicit proof of the fact that
${\rm equation\; (7.13)}\Leftrightarrow S_{\alpha,\beta}^1=0$ as an
example. By a similar study, we can show that $$  {\rm equation\;
(7.14)}\Leftrightarrow S_{\alpha,\beta}^2=0,\;\; {\rm equation\;
(7.15)}\Leftrightarrow S_{\alpha,\beta}^3=0, \;\;{\rm equation\;
(7.16)}\Leftrightarrow S_{\alpha,\beta}^4=0.$$ In fact, set
$r=\sum\limits_i u_i\otimes v_i$. Substituting
$$\alpha(x)=\sum\limits_iu_i\circ x\otimes v_i-u_i\otimes x\succ v_i,\;
\beta(x)=\sum\limits_iu_i\otimes x\circ v_i-u_i\prec x\otimes
v_i,\;\forall x\in A,$$ into equation (7.3) and after rearranging
the terms suitably, we divide $S_{\alpha,\beta}^1$ into three parts:
$$S_{\alpha,\beta}^1=(S1)+(S2)+(S3),$$
where
\begin{eqnarray*}
(S1)&=&\sum_{i,j}\{u_i\circ u_j\otimes v_i\otimes x\circ
v_j-u_i\otimes u_j\succ v_i\otimes x\circ v_j+u_i\otimes u_j\circ
v_i\otimes x\circ v_j\\
& &-u_i\prec u_j\otimes v_i\otimes x\circ v_j+v_i\otimes u_i\circ
u_j\otimes x\circ v_j-u_j\succ v_i\otimes
u_i\otimes x\circ v_j\\
& &+u_j\circ v_i\otimes u_i\otimes x\circ v_j-v_i\otimes u_i\prec
u_j\otimes x\circ v_j-u_j\otimes u_i\otimes(x\circ v_j)\circ
v_i\\&&-u_i\otimes u_j\otimes(x\circ
v_j)\circ v_i\}.\\
(S2)&=&\sum_{i,j}\{u_i\otimes(u_j\prec x)\succ v_i\otimes
v_j-u_i\otimes(u_j\prec x)\circ v_i\otimes v_j-v_i\otimes
u_i\circ(u_j\prec x)\otimes v_j\\ & &+v_i\otimes u_i\prec(u_j\prec
x)\otimes v_j+u_j\otimes u_i\prec(x\circ v_j)\otimes v_i+u_i\otimes
u_j\prec x\otimes v_j\circ v_i\\&&-u_i\prec v_j\otimes u_j\prec
x\otimes v_i\},\\
(S3)&=&\sum_{i,j}\{-u_i\circ(u_j\prec x)\otimes v_i\otimes
v_j+u_i\prec(u_j\prec x)\otimes v_i\otimes v_j+(u_j\prec x)\succ
v_i\otimes u_i\otimes v_j\\& &-(u_j\prec x)\circ v_i\otimes
u_i\otimes v_j+u_j\prec x\otimes u_i\otimes v_j\circ v_i-u_j\prec
x\otimes u_i\prec v_j\otimes v_i\\&&+u_i\prec(x\circ v_j)\otimes
u_j\otimes v_i\}.
\end{eqnarray*}
Since $r$ is symmetric and by Remark 2.6, we have
$$(S1)=-(1\otimes
1\otimes\mathfrak{l}_{\circ}(x))PA_3,\;\;(S2)=(1\otimes\mathfrak{r}_{\prec}(x)\otimes
1)PA_3^2,\;\;(S3)=(\mathfrak{r}_{\prec}(x)\otimes 1\otimes
1)PA_3^1.$$ Therefore the conclusion holds.
\end{proof}

\begin{theorem}
 Let $(A,\prec,\succ)$ be a
pre-alternative algebra and $r\in A\otimes A$ be symmetric. Let
$\alpha,\beta:A\rightarrow A\otimes A$ be two linear maps given by
equations (7.1) and (7.2) respectively, where
$r_{\prec}=r_{\succ}=r$. Then $(A,\prec,\succ,\alpha,\beta)$ is a
pre-alternative bialgebra if and only if equations (7.13)-(7.16) are
satisfied.\end{theorem}

\begin{proof} It follows from Theorem 7.3 and Proposition
7.6.\end{proof}

Next we give a Drinfeld's ``double" construction (\cite{CP}) for a
pre-alternative bialgebra.

\begin{theorem} Let $(A,\prec,\succ,\alpha,\beta)$ be a pre-alternative
bialgebra. Then there is a canonical pre-alternative bialgebra
structure on $A\oplus A^*$ such that both the inclusions
$i_1:A\rightarrow A\oplus A^*$ and $i_2:A^*\rightarrow A\oplus A^*$
into the two summands are homomorphisms of pre-alternative
bialgebras, where the pre-alternative bialgebra structure on $A^*$
is given in Example 6.6.
\end{theorem}

\begin{proof} Denote the pre-alternative algebra structure on $A^*$ induced by $\alpha^*$
and $\beta^*$ by $\prec_*$ and $\succ_*$ respectively, and its
associated alternative algebra structure by $\ast$. Let $r\in
A\otimes A^*\subset(A\oplus A^*)\otimes(A\oplus A^*)$ correspond to
the identity map ${\rm id}:A\rightarrow A$. Then the pre-alternative
algebra structure $\prec_{\bullet},\succ_{\bullet}$ on $A\oplus A^*$
is given by
$\mathcal{PAD}(A)=A\bowtie_{-\mathfrak{r}_{\succ}^*,\mathfrak{l}_{\circ}^*,
\mathfrak{r}_{\circ}^*,-\mathfrak{l}_{\prec}^*}^{-\mathfrak{r}_{\succ_*}^*,
\mathfrak{l}_{\ast}^*,\mathfrak{r}_{\ast}^*,-\mathfrak{l}_{\prec_*}^*}A^*$,
that is
\begin{eqnarray*} &&x\prec_{\bullet}y=x\prec y,\quad
x\succ_{\bullet}y=x\succ y,\quad
a^*\prec_{\bullet}b^*=a^*\prec_*b^*,\quad
a^*\succ_{\bullet}b^*=a^*\succ_*b^*,\\
&&x\prec_{\bullet}a^*=-\mathfrak{r}_{\succ}^*(x)b^*
+\mathfrak{l}_*^*(a^*)x,\quad
x\succ_{\bullet}a^*=\mathfrak{r}_{\circ}^*(x)a^*-\mathfrak{l}_{\prec_*}^*(a^*)x,
\\
&&a^*\prec_{\bullet}x=-\mathfrak{r}_{\succ_*}^*(a^*)x+\mathfrak{l}_{\circ}^*(x)a^*,\quad
a^*\succ_{\bullet}x=\mathfrak{r}_*^*(a^*)x-\mathfrak{l}_{\prec}^*(x)a^*,\;\;\forall
x,y\in A,a^*,b^*\in A^*. \end{eqnarray*} We shall denote its
associated alternative algebra structure by $\bullet$. Let
$\{e_i,...,e_n\}$ be a basis of $A$ and $\{e_1^*,...e_n^*\}$ be its
dual basis. Then $r=\sum\limits_ie_i\otimes e_i^*$. Next we prove
that
$$\alpha_{\mathcal{PAD}}(u)=(\mathfrak{r}_{\circ}(u)\otimes
1-1\otimes\mathfrak{l}_{\succ}(u))r,\quad {\rm and}\quad
 \beta_{\mathcal{PAD}}(u)=(1\otimes\mathfrak{l}_{\circ}(u)-\mathfrak{r}_{\prec}(u)\otimes
 1)r,$$
induces a (coboundary) pre-alternative bialgebra structure on
$A\oplus A^*$. Since $r$ is not symmetric we cannot apply Theorem
7.7 directly. Instead, We shall prove that $\alpha_{\mathcal{PAD}}$
and $\beta_{\mathcal{PAD}}$ satisfy equations (6.1)-(6.8) and the
conditions of Lemma 7.4. For the former, we give an explicit proof
of the fact that $\alpha_{\mathcal{PAD}}$ and
$\beta_{\mathcal{PAD}}$
 satisfy equation (6.3) as an example (the proof of the others is similar). In fact, we only need to prove
 $$(\mathfrak{r}_{\circ}(y)\otimes
 1-1\otimes\mathfrak{l}_{\succ}(y))(\mathfrak{l}_{\circ}(x)\otimes
 1-1\otimes\mathfrak{r}_{\prec}(x))(r-\tau r)=0,\forall x\in A.$$
We can prove the above equation in the following cases: (I) $x,y\in
A$; (II) $x,y\in A^*$; (III) $x\in A,y\in A^*$ and (IV) $x\in
A^*,y\in A$. As an example, we give an explicit proof of the first
case (the proof of the others is similar). Let $x=e_i,y=e_j$, then
the above equation becomes
\begin{eqnarray*}
& &\sum_k\{(e_i\bullet e_k^*)\bullet e_j\otimes e_k-e_k^*\bullet
e_j\otimes e_k\prec_{\bullet}e_i-e_i\bullet e_k^*\otimes
e_j\succ_{\bullet}e_k+e_k^*\otimes
e_j\succ_{\bullet}(e_k\prec_{\bullet}e_i)\}\\
 (*) &=&\sum_k\{(e_i\bullet e_k)\bullet e_j\otimes e_k^*-e_k\bullet
e_j\otimes e_k^*\prec_{\bullet}e_i-e_i\bullet e_k\otimes
e_j\succ_{\bullet}e_k^*+e_k\otimes
e_j\succ_{\bullet}(e_k^*\prec_{\bullet}e_i)\}.
\end{eqnarray*}
The coefficient of $e_m\otimes e_n$ on the left hand side of the
above equation $(*)$ is
\begin{eqnarray*}
& &\sum_k\{\langle (e_i\bullet e_n^*)\bullet e_j,e_m^*\rangle
-\langle e_k^*\bullet e_j,e_m^*\rangle \langle
e_k\prec_{\bullet}e_i,e_n^*\rangle -\langle e_i\bullet
e_k^*,e_m^*\rangle \langle e_j\succ_{\bullet}e_k,e_n^*\rangle \}\\
&=&\sum_k\{\langle e_n^*,e_k\prec e_i\rangle \langle
e_j,e_m^*\prec_*e_k^*\rangle +\langle e_i,e_n^*\succ_*e_k^*\rangle
\langle e_k\circ e_j,e_m^*\rangle -\langle
e_j,e_m^*\prec_*e_k^*\rangle \langle e_k\prec
e_i,e_n^*\rangle \\
& &-\langle e_i,e_k^*\succ_*e_m^*\rangle \langle e_j\succ
e_k,e_n^*\rangle \}
\end{eqnarray*}
The coefficient of $e_m\otimes e_n$ on the right hand side of the
equation $(*)$ is
\begin{eqnarray*}
& &\sum_k\{-\langle e_k\bullet e_j,e_m^*\rangle \langle
e_k^*\prec_{\bullet} e_i,e_n^*\rangle -\langle e_i\bullet
e_k,e_m^*\rangle \langle e_j\succ_{\bullet}e_k^*,e_n^*\rangle +
\langle e_j\succ_{\bullet}(e_m^*\prec_{\bullet}e_i),e_n^*\rangle \}\\
&=&\sum_k\{\langle e_k\circ e_j,e_m^*\rangle \langle
e_i,e_n^*\succ_*e_k^*\rangle +\langle e_i\circ e_k,e_m^*\rangle
\langle e_j,e_k^*\prec_*e_n^*\rangle -\langle e_j\succ
e_k,e_n^*\rangle \langle e_i,e_k^*\succ_*e_m^*\rangle \\
& &-\langle e_j,e_k^*\prec_*e_n^*\rangle \langle e_m^*,e_i\circ
e_k\rangle \}.
\end{eqnarray*}
So the coefficients of $e_m\otimes e_n$ on the both sides of the
equation $(*)$ are the same. Similarly, the coefficients of
$e_m^*\otimes e_n,e_m\otimes e_n^*$ and $e_m^*\otimes e_n^*$ on both
sides of the equation $(*)$ are the same.

On the other hand, we prove that
$S_{\alpha_{\mathcal{PAD}},\beta_{\mathcal{PAD}}}^i=0,i=1,2,3,4$. As
an example we give an explicit proof of the fact that
$S_{\alpha_{\mathcal{PAD}},\beta_{\mathcal{PAD}}}^1=0$.  In fact,
the coefficient of $e_m\otimes e_n\otimes e_p$ in
$S_{\alpha_{\mathcal{PAD}},\beta_{\mathcal{PAD}}}^1(e_k)$ is
\begin{eqnarray*}
& &-\langle e_j\succ_{\bullet}e_m^*,e_n^*\rangle \langle e_k\bullet
e_j^*,e_p^*\rangle +\langle e_j\bullet e_m^*,e_n^*\rangle \langle
e_k\bullet e_j^*,e_p^*\rangle -\langle
e_j\succ_{\bullet}e_n^*,e_m^*\rangle \langle e_k\bullet
e_j^*,e_p^*\rangle \\
& &+\langle e_j\bullet e_n^*,e_m^*\rangle \langle e_k\bullet
e_j^*,e_p^*\rangle -\langle (e_k\bullet
e_m^*)\bullet e_n^*,e_p^*\rangle -\langle (e_k\bullet e_n^*)\bullet e_m^*,e_p^*\rangle \\
&=&\langle e_j,e_m^*\prec_*e_n^*\rangle \langle
e_k,e_j^*\succ_*e_p^*\rangle +\langle e_j,e_m^*\succ_*
e_n^*\rangle \langle e_k,e_j^*\succ_*e_p^*\rangle +\langle e_j,e_n^*\prec_*e_m^*\rangle \langle e_k,e_j^*\succ_*e_p^*\rangle \\
& &+\langle e_j,e_n^*\succ_*e_m^*\rangle \langle
e_k,e_j^*\succ_*e_p^*\rangle -\langle e_k\bullet(e_m^*\ast
e_n^*+e_n^*\ast e_m^*),e_p^*\rangle \\
&=&\langle e_k,(e_m^*\ast e_n^*+e_n^*\ast
e_m^*)\succ_{\bullet}e_p^*-(e_m^*\ast e_n^*+e_n^*\ast
e_m^*)\succ_{\bullet}e_p^*\rangle =0.
\end{eqnarray*}
Similarly, the coefficients of $e_m^*\otimes e_n\otimes
e_p,e_m\otimes e_n^*\otimes e_p,e_m^*\otimes e_n^*\otimes e_p,
e_m\otimes e_n\otimes e_p^*,e_m^*\otimes e_n\otimes e_p^*,e_m\otimes
e_n^*\otimes e_p^*$ and $e_m^*\otimes e_n^*\otimes e_p^*$ in
$S_{\alpha_{\mathcal{PAD}},\beta_{\mathcal{PAD}}}^1(e_k)$ are all
zero. Therefore,
$S_{\alpha_{\mathcal{PAD}},\beta_{\mathcal{PAD}}}^1(e_k)=0$. By a
similar study,
$S_{\alpha_{\mathcal{PAD}},\beta_{\mathcal{PAD}}}^1(e_k^*)=0$. Hence
$\mathcal{PAD}(A)$ is a pre-alternative bialgebra.

 For $e_i\in A$, we have
\begin{eqnarray*}
\alpha_{\mathcal{PAD}}(e_i)&=&\sum_je_j\circ e_i\otimes
e_j^*-e_j\otimes e_i\succ_{\bullet}e_j^*\\&=& \sum_{j,m}e_j\circ
e_i\otimes e_j^*-e_j\otimes
e_m^*\langle e_j^*,e_m\circ e_i\rangle +e_j\otimes e_m\langle e_i,e_j^*\prec_*e_m^*\rangle \\
&=&\sum_{j,m}\langle e_i,e_j^*\prec_*e_m^*\rangle e_j\otimes
e_m=\alpha(e_i).
\end{eqnarray*}
Similarly we have $\beta_{\mathcal{PAD}}(e_i)=\beta(e_i)$, so the
inclusion $i_1: A\rightarrow A\oplus A^*$ is a homomorphism of
pre-alternative bialgebras. Similarly, the inclusion $i_2:
A^*\rightarrow A\oplus A^*$ is also a homomorphism of
pre-alternative bialgebras.\end{proof}

\begin{defn}{\rm Let $(A,\prec,\succ,\alpha,\beta)$ be a pre-alternative
bialgebra. With the pre-alternative bialgebra structure given in
Theorem 7.8, $A\oplus A^*$ is called a {\it Drinfeld's symplectic
double} of $A$. We denoted it by $\mathcal{PAD}(A)$.}
\end{defn}

\begin{prop}
 Let $(A,\prec,\succ,\alpha,\beta)$ be a
pre-alternative bialgebra with $\alpha,\beta$ defined by equations
(7.1) and (7.2), where $r_{\prec}=r_{\succ}=r\in A\otimes A$ and $r$
is a solution of $PA$-equations. Then $T_r$ is a homomorphism of
pre-alternative bialgebras from the pre-alternative bialgebra given
in Example 6.6 to $(A,\prec,\succ,\alpha,\beta)$.
\end{prop}

\begin{proof}
Note that $(1\otimes\alpha)r=r_{12}\prec r_{13}$ and $(\alpha\otimes
1)r=r_{13}\prec r_{23}$. Denote the pre-alternative algebra
structure on $A^*$ induced by $\alpha^*$ and $\beta^*$ by $\prec_*$
and $\succ_*$ respectively and the pre-alternative algebra structure
$\prec,\succ$ on $A$ is defined by the linear maps
$\gamma^*,\delta^*:A\otimes A\rightarrow A$ respectively. Then

{\small
$$T_r(a^*\prec_*b^*)=\langle 1\otimes(a^*\prec_*b^*),r\rangle
=\langle 1\otimes a^*\otimes b^*,(1\otimes\alpha)r\rangle =\langle
1\otimes a^*\otimes b^*,r_{12}\prec r_{13}\rangle =T_r(a^*)\prec
T_r(b^*),$$
$$(T_r\otimes T_r)\gamma(a^*)=\langle 1\otimes 1\otimes
a^*,r_{13}\prec r_{23}\rangle =(1\otimes 1\otimes a^*)(\alpha\otimes
1)r=\alpha(T_r(a^*)),$$} where $a^*,b^*\in A^*$. Similarly we have
$T_r(a^*\succ_*b^*)=T_r(a^*)\succ T_r(b^*)$ and $(T_r\otimes
T_r)\delta(a^*)=\beta(T_r(a^*))$. So the conclusion
holds.\end{proof}

\section{$PA$-equations and their properties}
\setcounter{equation}{0}
\renewcommand{\theequation}
{8.\arabic{equation}}

The simplest way to satisfy the conditions of Theorem 7.7 is given
as follows.

\begin{prop} Let $(A,\prec,\succ)$ be a pre-alternative algebra and $r\in
A\otimes A$ be symmetric. Let $\alpha,\beta:A\rightarrow A\otimes A$
be two linear maps defined by equations (7.1) and (7.2). Then
$(A,\prec,\succ,\alpha,\beta)$ is a pre-alternative bialgebra if $r$
satisfies $PA$-equations.
\end{prop}

\begin{prop} Let $(A,\prec,\succ)$ be a pre-alternative algebra and $(As(A),\circ)$ be the
associated alternative algebra. Let $r\in A\otimes A$ be a symmetric
solution of $PA$-equations in $A$. Then the pre-alternative algebra
structure ``$\prec_{\bullet},\succ_{\bullet}$'' on the Drinfeld's
symplectic double $\mathcal{PAD}(A)$ is given as follows:
\begin{equation}
a^*\prec_{\bullet}b^*=a^*\prec_*b^*=
\mathfrak{l}_{\circ}^*(T_r(b^*))a^*-\mathfrak{r}_{\succ}^*(T_r(a^*))b^*,\end{equation}\begin{equation}
a^*\succ_{\bullet}b^*=a^*\succ_*b^*=
\mathfrak{r}_{\circ}^*(T_r(a^*))b^*-\mathfrak{l}_{\prec}^*(T_r(b^*))a^*,\end{equation}\begin{equation}
x\prec_{\bullet}a^*=x\prec
T_r(a^*)+T_r(\mathfrak{r}_{\succ}^*(x)a^*)-\mathfrak{r}_{\succ}^*(x)a^*,\end{equation}\begin{equation}
x\succ_{\bullet}a^*=\mathfrak{r}_{\circ}^*(x)a^*-T_r(\mathfrak{r}_{\circ}^*(x)a^*)
+x\succ
T_r(a^*),\end{equation}\begin{equation}a^*\prec_{\bullet}x=-T_r(\mathfrak{l}_{\circ}^*(x)a^*)+T_r(a^*)\prec
x+\mathfrak{l}_{\circ}^*(x)a^*,\end{equation}\begin{equation}
a^*\succ_{\bullet}x=T_r(a^*)\succ
x+T_r(\mathfrak{l}_{\prec}^*(x)a^*)-\mathfrak{l}_{\prec}^*(x)a^*,\end{equation}
where $x\in A,a^*,b^*\in A^*$, the pre-alternative algebra structure
on $A^*$ is denoted by ``$\prec_*,\succ_*$" and the associated
alternative algebra structure on $As(A^*)$ is denoted by ``$\ast$''.
\end{prop}

\begin{proof}
Let $\{e_1,...,e_n\}$ be a basis of $A$ and $\{e_1^*,...,e_n^*\}$ be
its dual basis. Suppose that
$$e_i\prec
e_j=\sum_{i,j}c_{ij}^ke_k,\quad e_i\succ
e_j=\sum_{i,j}d_{ij}^ke_k,\quad r=\sum_{i,j}a_{ij}e_i\otimes
e_j,\quad a_{ij}=a_{ji}.$$ Then $T_r(e_i^*)=\sum\limits_ka_{ki}e_k$.
Thus (for any $k,l$)
\begin{eqnarray*}
e_k^*\prec_{\ast}e_l^*&=&\sum_s\{\langle e_k^*\otimes
e_l^*,\alpha(e_s)\rangle e_s^*\}
=\sum_{i,s}\{a_{il}(c_{is}^k+d_{is}^k)-a_{ki}d_{si}^l\}e_s^*
\\&=&\sum_{i,s}\{a_{il}\langle e_i\circ e_s,e_k^*\rangle -e_{ki}\langle
e_s\succ e_i,e_l^*\rangle \}e_s^*=
\mathfrak{l}_{\circ}^*(T_r(e_l^*))e_k^*-\mathfrak{r}_{\succ}^*(T_r(e_k^*))e_l^*.
\end{eqnarray*}
So equation (8.1) holds. Similarly equation (8.2) holds. Therefore
\begin{eqnarray*}\mathfrak{l}_{\prec_*}^*(e_k^*)e_m&=&\sum_s\langle
e_m,e_k^*\prec_*e_s^*\rangle e_s =\sum_s\langle
e_m,\mathfrak{l}_{\circ}^*(T_r(e_s^*))e_k^*-
\mathfrak{r}_{\succ}^*(T_r(e_k^*))e_s^*\rangle e_s
\\&=&\sum_s\langle T_r(e_s^*)\circ e_m,e_k^*\rangle e_s-\langle
e_m\succ T_r(e_k^*),e_s^*\rangle e_s
=T_r(\mathfrak{r}_{\circ}^*(e_m)e_k^*)-e_m\succ T_r(e_k^*).
\end{eqnarray*}
Therefore equation (8.4) follows from the fact that
$e_m\succ_{\bullet}e_k^*=\mathfrak{r}_{\circ}^*(e_m)e_k^*-\mathfrak{l}_{\prec_*}^*(e_k^*)e_m$.
Similarly, we can get the other equations.
\end{proof}

\begin{prop} Let $(A,\prec,\succ)$ be a pre-alternative algebra and
$(As(A),\circ)$ be the associated alternative algebra. Let $r\in
A\otimes A$ be symmetric. Then $r$ is a solution of one of
$PA_j^i$-equations ($i=1,2,j=1,2,3$) if and only if $T_r$ satisfies
\begin{equation}
T_r(a^*)\circ T_r(b^*)=T_r(\mathfrak{r}_{\prec}^*(T_r(a^*))b^*+
\mathfrak{l}_{\succ}^*(T_r(b^*))a^*),\;\; \forall a^*,b^*\in
A^*,\end{equation} that is, $T_r$ is an
 $Al$-operator of $As(A)$ associated to the bimodule
 $(A^*,\mathfrak{r}_{\prec}^*,\mathfrak{l}_{\succ}^*)$.
So in this case $PA_j^i$-equations ($i=1,2,j=1,2,3$) are all
equivalent. Moreover, if $r$ is a solution of one of
$PA_j^i$-equations ($i=1,2,j=1,2,3$), then there is a
 pre-alternative algebra structure on $A^*$ given by \begin{equation}a^*\prec
 b^*=\mathfrak{l}_{\succ}^*(T_r(b^*))a^*,\quad a^*\succ
 b^*=\mathfrak{r}_{\prec}^*(T_r(a^*))b^*,\quad \forall a^*,b^*\in
 A^*.\end{equation}
The associated alternative algebra structure $As(A^*)$ is the same
as the one given by equations (8.1) and (8.2)
 which is induced by $r$ in the sense of coboundary
 pre-alternative bialgebras.
\end{prop}

\begin{proof} It follows from a similar proof as of Proposition 3.6.
\end{proof}

\begin{defn}
{\rm  Let $(A,\prec,\succ)$ be a pre-alternative algebra. A bilinear
form $\mathcal{B}:A\otimes A\rightarrow \textbf{k}$ is called a {\it
2-cocycle} of $A$ if
\begin{equation}\mathcal{B}(x\circ y,z)=\mathcal{B}(x,y\succ
z)+\mathcal{B}(y,z\prec x),\quad \forall x,y,z\in A.\end{equation}}
\end{defn}

\begin{prop}  Let $(A,\prec,\succ)$ be a pre-alternative algebra and
$(As(A),\circ)$ be the
   associated alternative algebra. Let $\mathcal B$ be a 2-cocycle
   of $(A,\prec,\succ)$. Then the bilinear form $\omega$ defined by
\begin{equation}
\omega(x,y)=\mathcal{B}(x,y)-\mathcal{B}(y,x),\quad \forall x,y\in
A,\end{equation} is a closed form on $As(A)$.
\end{prop}

\begin{proof} Straightforward.\end{proof}

\begin{prop}
 Let $(A,\prec,\succ)$ be a
pre-alternative algebra and $r\in A\otimes A$. Suppose that $r$ is
symmetric and nondegenerate. Then $r$ is a solution of one of
$PA_j^i$-equations ($i=1,2,j=1,2,3$) in $(A,\prec,\succ)$ if and
only if the (nondegenerate) bilinear form $\mathcal B$ induced by
$r$ through equation (1.10) is a 2-cocycle of
$(A,\prec,\succ)$.\end{prop}

\begin{proof}
 Let
$r=\sum_ia_i\otimes b_i$. Since $r$ is symmetric, we have
$\sum_ia_i\otimes b_i=\sum_ib_i\otimes a_i$. Therefore
$T_r(v^*)=\sum_i\langle v^*,a_i\rangle b_i=\sum_i\langle
v^*,b_i\rangle a_i$ for any $v^*\in A^*$. Since $r$ is
nondegenerate, for any $x,y,z\in A$, there exist $u^*,v^*,w^*\in
A^*$ such that $x=T_r(u^*),y=T_r(v^*),z=T_r(w^*)$. Therefore

{\small \begin{eqnarray*} \mathcal{B}(x,z\circ y)&=&\langle
u^*,T_r(w^*)\circ T_r(v^*)\rangle =\sum_{i,j}\langle w^*,b_i\rangle
\langle v^*,b_j\rangle \langle
u^*,a_i\circ a_j\rangle=\langle u^*\otimes v^*\otimes w^*,r_{13}\circ r_{12}\rangle, \\
\mathcal{B}(y,x\prec z\rangle &=&\langle v^*,T_r(u^*)\prec
T_r(w^*)\rangle =\sum_{i,j}\langle u^*,b_i\rangle \langle
w^*,b_j\rangle \langle v^*,a_i\prec a_j\rangle =\langle u^*\otimes v^*\otimes w^*,r_{12}\prec r_{23}\rangle, \\
\mathcal{B}(y\succ x,z)&=&\langle T_r(v^*)\succ T_r(u^*),w^*\rangle
=\sum_{i,j}\langle v^*,b_i\rangle \langle u^*,b_j\rangle \langle
w^*,a_i\succ a_j\rangle=\langle u^*\otimes v^*\otimes
w^*,r_{23}\succ r_{13}\rangle.
\end{eqnarray*}}
Hence $\mathcal B$ is a 2-cocycle of $(A,\prec,\succ)$ if and only
if equation (7.8) holds. By Proposition 8.3, the conclusion follows.
\end{proof}

\begin{coro} Let $(A,\prec,\succ)$ be a pre-alternative algebra
and $r\in A\otimes A$. Assume $r$ is symmetric and there exists a
nondegenerate symmetric bilinear form $h(x,y)$ on $A$ which is
associative in the sense that
\begin{equation}
h(x\prec y,z)=h(x,y\succ z),\;\;\forall x,y,z\in A.
\end{equation}
 Define a linear map $\varphi:A\rightarrow A^*$ by $\langle
\varphi(x),y\rangle =h(x,y)$. Then
$\tilde{T}_r=T_r\varphi:A\rightarrow A$ is an $Al$-operator
associated to the bimodule
$(A,\mathfrak{l}_{\succ},\mathfrak{r}_{\prec})$ if and only if $r$
is a symmetric solution of $PA$-equations. In this case,
$\tilde{T}_r$ satisfies the following equation:
\begin{equation}
\tilde{T}_r(x)\circ \tilde{T}_r(y)=\tilde{T}_r(\tilde{T}_r(x)\succ
y+x\prec \tilde{T}_r(y)).\end{equation} So we can define a
pre-alternative algebra structure on $A$ by \begin{equation}x\prec
y=x\prec \tilde{T}_r(y),\quad  x\succ y=\tilde{T}_r(x)\succ
y.\end{equation}
\end{coro}

\begin{proof}
It follows from a similar proof as of Corollary 3.7.\end{proof}

\begin{remark}{\rm In fact, a symmetric bilinear form on a pre-alternative algebra $(A,\prec,\succ)$ satisfying
equation (8.11) is a 2-cocycle of $(A,\prec,\succ)$.}
\end{remark}

By a similar proof as of Proposition 3.9, we have the following
conclusion:

\begin{prop}
 Let $(A,\circ)$ be an alternative algebra.
 Let $(V,L,R)$ be a bimodule of $A$ and
 $(V^*,R^*,L^*)$ be its dual bimodule. Suppose that $Al:V\rightarrow
 A$ is an $Al$-operator associated to $(V,L,R)$. Then $r=Al+\tau
 Al$ is a symmetric solution of $PA$-equations in $Al(V)\ltimes_{0,L^*,R^*,0}V^*$,
 where $Al(V)\subset A$ is a pre-alternative algebra given by equation
 (2.12) and  $(V^*,0,L^*,R^*,0)$ is a bimodule of $Al(V)$.
 \end{prop}

\begin{coro}
Let $(A,\prec,\succ)$ be a pre-alternative algebra. Then
\begin{equation}
r=\sum_i(e_i\otimes e_i^*+e_i^*\otimes e_i)\end{equation} is a
symmetric solution of $PA$-equations in
$A\ltimes_{0,\mathfrak{l}_{\succ}^*,\mathfrak{r}_{\prec}^*,0}A^*$,
where $\{e_1,...,e_n\}$ is a basis of $A$ and $\{e_1^*,...,e_n^*\}$
is its dual basis. Moreover $r$ is nondegenerate and the induced
$2$-cocycle $\mathcal{B}$ of
$A\ltimes_{0,\mathfrak{l}_{\succ}^*,\mathfrak{r}_{\prec}^*,0}A^*$ is
given by equation (3.2).
\end{coro}

\begin{proof} It follows by
letting $V=A, (L,R)=(\mathfrak{l}_{\succ},\mathfrak{r}_{\prec})$ and
$Al={\rm id}$ in Proposition 8.9.\end{proof}

\begin{coro}
Let $(A,\prec,\succ)$ be a pre-alternative algebra and
$(As(A),\circ)$ be the associated alternative algebra. If $r$ is a
nondegenerate symmetric solution of $PA$-equations in $A$, then
there is a new compatible pre-alternative algebra structure on
$As(A)$ given by
\begin{equation}
x\prec'y=T_r(\mathfrak{l}_{\succ}^*(y)T_r^{-1}(x)),\quad x\succ'y=
T_r(\mathfrak{r}_{\prec}^*(x)T_r^{-1}(y)),\;\; \forall x,y\in
A,\end{equation} which is just the pre-alternative algebra structure
given by
\begin{equation}\mathcal{B}(x\prec'y,z)=\mathcal{B}(x,y\succ
z),\quad \mathcal{B}(x\succ'y,z)=\mathcal{B}(y,z\prec x),\;\;
\forall x,y,z\in A,\end{equation} where $\mathcal B$ is the
nondegenerate symmetric 2-cocycle of $A$ induced by $r$  through
equation (1.10).
\end{coro}

\begin{prop}
 Let $(A,\prec,\succ,\alpha,\beta)$ be a
 pre-alternative bialgebra arising from a symmetric solution $r$ of
 $PA$-equations and  $A^*,-\mathfrak{r}_{\succ}^*,\mathfrak{l}_{\circ}^*,
 \mathfrak{r}_{\circ}^*,-\mathfrak{l}_{\prec}^*,-\mathfrak{r}_{\succ_*}^*,
 \mathfrak{l}_{\ast}^*,\mathfrak{r}_{\ast}^*,-\mathfrak{l}_{\prec_*}^*)$
 be the corresponding matched pair of pre-alternative algebras.

 ${\rm (1)}$ $A\bowtie_{-\mathfrak{r}_{\succ}^*,\mathfrak{l}_{\circ}^*,
 \mathfrak{r}_{\circ}^*,-\mathfrak{l}_{\prec}^*}^{-\mathfrak{r}_{\succ_*}^*,
 \mathfrak{l}_{\ast}^*,\mathfrak{r}_{\ast}^*,-\mathfrak{l}_{\prec_*}^*}A^*$
 is isomorphic to
 $A\ltimes_{-\mathfrak{r}_{\succ}^*,\mathfrak{l}_{\circ}^*,
 \mathfrak{r}_{\circ}^*,-\mathfrak{l}_{\prec}^*}A^*$
 as pre-alternative algebras.

 ${\rm (2)}$ The symmetric solutions of $PA$-equations are in one-to-one correspondence
 with linear maps $T_r:A^*\rightarrow A$ whose graphs are Lagrangian
  pre-alternative subalgebras (with respect to the bilinear form $(3.11)$) of  $A\ltimes_{-\mathfrak{r}_{\succ}^*,\mathfrak{l}_{\circ}^*,
 \mathfrak{r}_{\circ}^*,-\mathfrak{l}_{\prec}^*}A^*$.
\end{prop}

\begin{proof} It follows from a similar proof as of Proposition
3.13. \end{proof}

\section{Comparison between alternative D-bialgebras and pre-alternative bialgebras}

The results in the previous sections allow us to compare the
alternative D-bialgebras (see Appendix A) and pre-alternative
bialgebras in terms of the following properties:
 matched pairs of alternative algebras, alternative algebra structures on the direct sum of the
alternative algebras in the matched pairs, bilinear forms on the
direct sum of the alternative algebras in the matched pairs, double
structures on the direct sum of the alternative algebras in the
matched pairs, algebraic equations associated to coboundary cases,
nondegenerate solutions, $Al$-operators of alternative algebras and
constructions from pre-alternative algebras. We list them in Table
1.

\begin{table}[t]\caption{Comparison between alternative D-bialgebras and pre-alternative bialgebras}
\begin{tabular}{|c|c|c|}
\hline Algebras & Alternative D-bialgebras & Pre-alternative
bialgebras\\\hline Matched pairs of alternative  &
$(A,A^*,\mathfrak{r}_{\circ}^*,
\mathfrak{l}_{\circ}^*,\mathfrak{r}_*^*,\mathfrak{l}_*^*)$ &
$(As(A),As(A^*),\mathfrak{r}_{\prec}^*,\mathfrak{l}_{\succ}^*$,
\\ algebras&& $\mathfrak{r}_{\prec_*}^*,
\mathfrak{l}_{\succ_*}^*)$\\\hline
Alternative algebra  structures on   & alternative & phase spaces\\
the direct sum of the alternative  & analogues of Manin&\\
algebras in the matched pairs&triples&\\\hline Bilinear forms on &
symmetric & skew-symmetric
\\\cline{2-3} the direct sum of the alternative & $\langle x+a^*,y+b^*\rangle  $ & $\langle x+a^*,y+b^*\rangle  $\\algebras in the
matched pairs &$=\langle x,b^*\rangle  +\langle a^*,y\rangle  $ &
$=-\langle x,b^*\rangle  +\langle a^*,y\rangle  $\\\cline{2-3}
&invariant & symplectic forms\\\hline Double structures on  &
Drinfeld's doubles & Drinfeld's symplectic \\ the direct sum of the
alternative&&doubles\\algebras in the matched pairs&&\\\hline
Algebraic equations associated & skew-symmetric solutions &
symmetric solutions
\\\cline{2-3}
to coboundary cases& alternative YBEs in & $PA$-equations in
\\&alternative algebras&pre-alternative algebras\\\hline Nondegenerate
solutions & symplectic forms of   & 2-cocycles of pre-alternative\\
&alternative algebras&algebras\\\hline $Al$-operators & associated
to $(\mathfrak{r}_{\circ}^*,\mathfrak{l}_{\circ}^*)$ & associated to
$(\mathfrak{r}_{\prec}^*,\mathfrak{l}_{\succ}^*)$\\\cline{2-3} &
\enspace skew-symmetric parts &\enspace symmetric parts
\\\hline
Constructions from & $r=\sum\limits_{i=1}^n (e_i\otimes
e_i^*-e_i^*\otimes e_i)$ & $r=\sum\limits_{i=1}^n (e_i\otimes
e_i^*+e_i^*\otimes e_i)$\\\cline{2-3} pre-alternative algebras &
induced bilinear
forms & induced bilinear forms\\
& $\langle x+a^*,y+b^*\rangle  $ & $\langle x+a^*,y+b^*\rangle  $\\
&$=-\langle x,b^*\rangle  +\langle a^*,y\rangle  $ & $=\langle
x,b^*\rangle  +\langle a^*,y\rangle  $\\\hline
\end{tabular}
\end{table}

From Table 1 we observe that there is a clear analogy between the
alternative D-bialgebras and pre-alternative bialgebras. Moreover,
due to the correspondences between certain symmetries and
skew-symmetries appearing in the table, we regard it as a kind of
duality.

\section*{Appendix: Another approach to alternative D-bialgebras}

\setcounter{equation}{0}
\renewcommand{\theequation}
{A.\arabic{equation}}

In this section we prove the main results of \cite{G} by using a
slightly different method (in fact, we will prove some results which
are a little stronger than \cite{G}). Moreover, there will be some
more results like  the ``Drinfeld's double" theorem for an
alternative bialgebra (Theorem A11) and a homomorphism property of
an alternative bialgebra (Theorem A12), which were not presented in
\cite{G}. We omit the proofs since they are quite similar to the
study of pre-alternative bialgebras.

\noindent{\bf Theorem A1} \quad {\it Let $(A,\circ)$ be an
alternative algebra and $(A^*,\star)$ be an alternative algebra
induced by a linear map $\Delta:A\rightarrow A\otimes A$. Then
$(A,\circ,\Delta)$ is an alternative $D$-bialgebras if and only if
$(A,A^*,\mathfrak{r}_{\circ}^*,\mathfrak{l}_{\circ}^*,\mathfrak{r}_{\star}^*,
\mathfrak{l}_{\star}^*)$ is a matched pair of alternative algebras.}

\noindent {\bf Theorem A2}\quad {\it Let $(A,\circ,\Delta)$ be an
alternative algebra equipped with a linear map $\Delta:A\rightarrow
A\otimes A$ such that $\Delta^*:A^*\otimes A^*\rightarrow A^*$
induces an alternative algebra structure on $A^*$. Then
$(A,A^*,\mathfrak{r}_{\circ},\mathfrak{l}_{\circ},\mathfrak{r}_{\star},
\mathfrak{l}_{\star})$ is a matched pair of alternative algebras if
and only if the following equations hold:
\begin{equation}\Delta(x\circ y)=(-\mathfrak{l}_{\circ}(x)\otimes 1+1\otimes
\mathfrak{Ass}_{\circ}(x))\Delta(y)+(\mathfrak{r}_{\circ}(y)\otimes
1)\Delta(x)+(\mathfrak{r}_{\circ}(y)\otimes
1-1\otimes\mathfrak{l}_{\circ}(y))\tau\Delta(x),\end{equation}
\begin{equation}\Delta(x\circ y)=(\mathfrak{Ass}_{\circ}(y)\otimes
1-1\otimes\mathfrak{r}_{\circ}(y))\Delta(x)+(1\otimes\mathfrak{l}_{\circ}(x))\Delta(y)
+(1\otimes\mathfrak{l}_{\circ}(x)-\mathfrak{r}_{\circ}(x)\otimes
1)\tau\Delta(y),\end{equation}\begin{equation}\Delta(x\circ y+y\circ
x)=(\mathfrak{r}_{\circ}(y)\otimes 1
+1\otimes\mathfrak{l}_{\circ}(y))\Delta(x)+(1\otimes\mathfrak{l}_{\circ}(x)+
\mathfrak{r}_{\circ}(x)\otimes
1)\Delta(y),\end{equation}\begin{equation}(\Delta+\tau\Delta)(x\circ
y)=(\mathfrak{r}_{\circ}(y)\otimes
1)\Delta(x)+(1\otimes\mathfrak{r}_{\circ}(y))\tau\Delta(x)+(\mathfrak{l}_{\circ}(x)\otimes
1)\tau\Delta(y)+(1\otimes\mathfrak{l}_{\circ}(x))\Delta(y),\end{equation}
where $x,y\in A$, the multiplication $\star$ is induced by $\Delta$
and
$\mathfrak{Ass}_{\circ}=\mathfrak{l}_{\circ}+\mathfrak{r}_{\circ}$.}

\noindent {\bf Remark A3}\quad  Note that equations (A.1) and (A.4)
have already appeared in \cite{G} (Lemma 2).

\noindent {\bf Definition A4}\quad An alternative bialgebra
$(A,\circ,\Delta)$ is called {\it coboundary} if there exists an
$r\in A\otimes A$ such that $\Delta$ is given by
\begin{equation}\Delta(x)=(\mathfrak{r}_{\circ}(x)\otimes
1-1\otimes\mathfrak{l}_{\circ}(x))r,\quad \forall x\in
A.\end{equation}

\noindent {\bf Remark A5}\quad Note that the above definition is as
the same as Definition 3.3.


\noindent {\bf Lemma A6}\quad {\it Let $A$ be a vector space. Then a
linear map $\Delta:A\rightarrow A\otimes A$ induces an alternative
algebra structure on $A^*$ if and only if the following two
equations hold: (for any $x,y\in A$)
\begin{equation}
(\Delta\otimes 1)\Delta(x)+(\tau\otimes 1)(\Delta\otimes
1)\Delta(x)=(1\otimes\Delta)\Delta(x)+(\tau\otimes
1)(1\otimes\Delta)\Delta(x),\end{equation}\begin{equation}
(\Delta\otimes 1)\Delta(x)+(1\otimes\tau)(\Delta\otimes
1)\Delta(x)=(1\otimes\Delta)\Delta(x)+(1\otimes\tau)(1\otimes\Delta)\Delta(x).\end{equation}}

\noindent {\bf Definition A7}\quad {\rm Let $(A,\circ)$ be an
alternative algebra.  The following equations are called {\it
$A_i$-equations} ($i=1,2$)  in $A$ respectively:
\begin{equation}A_1=r_{23}\circ r_{12}-r_{13}\circ r_{23}-r_{12}\circ
r_{13}=0,\end{equation}\begin{equation}A_2=r_{12}\circ
r_{23}-r_{23}\circ r_{13}-r_{13}\circ r_{12}=0.\end{equation}}

Note that $A_1$ given by equation (A.9) is exactly $C_A(r)$ given by
equation (3.5).

\noindent {\bf Proposition A8}\quad {\it Let $(A,\circ)$ be an
alternative algebra and $r\in A\otimes A$. If $r$ is skew-symmetric,
then $A_1$ and $A_2$  are the same. }

\noindent {\bf Proposition A9}\quad {\it Let $(A,\circ)$ be an
alternative algebra. Let $r\in A\otimes A$ be skew-symmetric. Define
a linear map $\Delta:A\rightarrow A\otimes A$ by equation (A.5).
Then $\Delta$ induces an alternative algebra structure on $A^*$ if
and only if the following equations hold:
\begin{equation} -(\mathfrak{r}_{\circ}(x)\otimes 1\otimes 1)A_1-(1\otimes
\mathfrak{r}_{\circ}(x)\otimes 1)A_2+(1\otimes
1\otimes\mathfrak{l}_{\circ}(x))(A_1+A_2)=0,\end{equation}
\begin{equation} -(\mathfrak{r}_{\circ}(x)\otimes 1\otimes 1)(A_1+A_2)+(1\otimes
\mathfrak{l}_{\circ}(x)\otimes 1)A_2+(1\otimes
1\otimes\mathfrak{l}_{\circ}(x))A_1=0,\end{equation} for any $x\in
A$.}

\noindent {\bf Theorem A10} (\cite{G}, Theorem 2)\quad{\it Let
$(A,\circ)$ be an alternative algebra and $r\in A\otimes A$. Let
$\Delta:A\rightarrow A$ be a linear map defined by equation (A.5).
If $r$ is a skew-symmetric solution of alternative Yang-Baxter
equation in $A$, then $(A,\circ,\Delta)$ is an alternative
bialgebra.}

\noindent {\bf Theorem A11}\quad {\it Let $(A,\circ,\Delta_A)$ be an
alternative bialgebra. Then there is a canonical alternative
bialgebra structure on $A\oplus A^*$ such that the inclusion
$i_1:A\rightarrow A\oplus A^*$ is a homomorphism of alternative
bialgebras, which the alternative bialgebra structure on $A$ is
given by $(A,\circ,-\Delta_A)$  and the inclusion
$i_2:A^*\rightarrow A\oplus A^*$ is a homomorphism of alternative
bialgebras, which the alternative bialgebra structure on $A^*$ is
given by $(A^*,\ast,\Delta_B)$, where ``$\ast$" is induced by
$\Delta_A$ and the alternative algebra structure ``$\circ$" on $A$
is induced by $\Delta_B:A^*\rightarrow A^*\otimes A^*$.}

\noindent {\bf Theorem A12}\quad {\it Let $(A,\circ,\Delta)$ be an
 alternative bialgebra arising from a solution $r$ of alternative Yang-Baxter equation in $A$.
 Then $T_r$ is a
homomorphism of alternative bialgebras from $(A^*,\ast,\delta)$ to
$(A,\circ,-\Delta)$, where ``$\ast$" is induced by $\Delta$ and the
alternative algebra structure ``$\circ$" on $A$ is induced by
$\delta:A^*\rightarrow A^*\otimes A^*$.}

\section*{Acknowledgements}

The authors thank Professor J.-L. Loday for useful suggestion. The
first author also thanks Professor V.N. Zhelyabin for kindly sending
his paper to him. This work was supported in part by the National
Natural Science Foundation of China (10621101), NKBRPC
(2006CB805905) and SRFDP (200800550015).


\begin{thebibliography}{}

\bibitem[AS1]{AS1} M. Aguiar and F. Sottile, {\it Structure
of the Loday-Ronco Hopf algegbra of trees}, J. Algebra \textbf{295}
(2006) 473-511.
\bibitem[AS2]{AS2} M. Aguiar and F. Sottile, {\it Cocommutative
Hopf algebras of permutations and trees}, J. of Algebraic
Combinatorics \textbf{22} (2005) 451-470.


\bibitem [Bai1]{Bai1} C.M. Bai, {\it A further study on non-abelian phase
spaces: Left-symmetric algebraic approach and related geometry},
Rev. Math. Phys. \textbf{18} (2006) 545-564.



\bibitem [Bai2]{Bai2} C.M. Bai, {\it A unified algebraic approach to classical Yang-Baxter equation},
J. Phy. A: Math. Theor. \textbf{40} (2007) 11073-11082.









\bibitem [Bax]{Bax} G.
Baxter, {\it An analytic problem whose solution follows from a
simple algebraic identity}, Pacific J. Math. \textbf{10} (1960)
731-742.

\bibitem[Bu]{Bu} D. Burde, {\it Left-symmetric algebras, or pre-Lie
algebras in geometry and physics}, Cent. Eur. J. Math. \textbf{4}
(2006) 323-357.

\bibitem[Cha]{Cha} F, Chapoton, {\it Un
th\'{e}or\`{e}me de Cartier-Milnor-Moore-Quillen pour big\`{e}bres
dendriformes et les alg\`{e}bres braces}, J. Pure Appl. Alg.
\textbf{168} (2002) 1-18.

\bibitem[CP]{CP} V. Chari, A. Pressley, {\it A guide to quantum
groups}, Cambridge University Press, Cambridge, (1994).
\bibitem[D]{D}
V.G. Drinfeld, {\it Hamiltonian structures on Lie groups, Lie
bialgebras and the geometric meaning of the classical Yang-Baxter
equations}, Soviet Math. Dokl. \textbf{27} (1983) 68-71.

\bibitem[EG]{EG} K. Ebrahimi-Fard and L. Guo, {\it Rota-Baxter algebras in
renormalization of perturbative quantum field theory}, in
``Universality and Renormalization", Fields Institute Communicatins
{\bf 50}, Amer. Math. Soc. (2007) 47-105.




\bibitem[Fra1]{Fra1}
A. Frabetti, {\it Dialgebra homology of associative algebras}, C.R.
Acad. Sci. Paris \textbf{325} (1997) 135-140.
\bibitem[Fra2]{Fra2} A. Frabetti,
{\it Leibniz homology of dialgebras of matrices,}  J. Pure Appl.
Alg. \textbf{129} (1998) 123-141.

\bibitem[G]{G} M.E. Goncharov, {\it The
classical Yang-Baxter equation on alternative algebras: the
alternative D-bialgebra structure on Cayley-Dickson matrix
algebras}, Siberian Math. J. \textbf{48} (2007) 809-823.
\bibitem[Ho]{Ho} R. Holtkamp, {\it On Hopf
algebra structure over operads}, arXiv: math. RA/0407074.
\bibitem[J]{J} N. Jacobson, {\it Structure and representations of Jordan
algebras}, Amer. Math. Soc. Providence, RI, (1968) (Amer. Math. Soc.
Colloq. Publ. \textbf{39}).
\bibitem[K1]{K1} B.A. Kupershmidt,
{\it Non-abelian phase space}, J. Phys. A: Math. Gen. \textbf{27}
(1994) 2801-2810.
\bibitem[K2]{K2} B.A. Kuperschmidt, {\it What a classical
$r$-matrix really is}, J. Nonlinear Math. Phy. \textbf{6} (1999)
448-488.
\bibitem[KS]{KS} E.N. Kuz'min and I.P. Shestakov, {\it Non-associative
structures}, Algebra VI, Encyclopaedia Math. Sci. \textbf{57},
Springer, Berlin, (1995)  197-280.
\bibitem[L1]{L1} J.-L. Loday, {\it Dialgebras}, in
``{\it Dialgebras and related operads}", Lecture Notes in Math.,
\textbf{1763}, Springer, Berlin, (2001) 7-66.

\bibitem[L2]{L2} J.-L. Loday,
{\it Arithmetree}, J. Algebra \textbf{258} (2002) 275-309.
\bibitem[L3]{L3} J.-L. Loday, {\it Scindement
d'associativit\'{e} et alg\`{e}bres de Hopf}, Proceeding of the
Conference in honor of Jean Leray, Nantes (2002), S\'{e}minaire et
Congr\`{e}s (SMF), \textbf{9} (2004) 155-172.

\bibitem[L4]{L4} J.-L. Loday,
{\it Generalized bialgebras and triples of operads}, Ast\'{e}risque
\textbf{320} (2008), vi+114 pp.


\bibitem[Ron]{Ron} M. Ronco,
{\it Eulerian idempotents and Milnor-Moore theorem for certain
non-commutative Hopf algebras}, J. Algebra \textbf{254} (2002)
152-172.


\bibitem [R] {R} G. -C. Rota,
{\it Baxter algebras and combinatorial identities I}, Bull. Amer.
Math. Soc. \textbf{75} (1969) 325-329.

\bibitem[S1]{S1} R.D. Schafer,
{\it Representations of Alternative algebras}, Trans. Amer. Math.
Soc. \textbf{72} (1952) 1-17.
\bibitem[S2]{S2} R.D. Schafer, {\it On the algebras formed by the Cayley-Dickson process}.
Amer. J. Math. \textbf{76} (1954) 435-446.
\bibitem [S3] {S3} R.D. Schafer, {\it An introduction to
nonassociative algebras}, Dover Publications Inc., New York (1995).
\bibitem[Z]{Z} V.N. Zhelyabin, {\it Jordan
bialgebras and their connections with Lie bialgebras}, Algebra and
Logic \textbf{36} (1997) 1-15.
\end{thebibliography}
\end{document}